\documentclass[a4paper,reqno]{amsart}

\usepackage[utf8]{inputenc} 
\usepackage{mathrsfs}
\usepackage{amsfonts,amsmath,amssymb,amsthm}

\usepackage{enumerate}
\usepackage{graphicx,tikz}
\usepackage{hyperref}
\usepackage[alphabetic,initials]{amsrefs}
\usetikzlibrary{decorations.pathreplacing}
\usetikzlibrary{arrows}

\newtheorem{thm}{Theorem}[section]

\newtheorem{lem}[thm]{Lemma}
\newtheorem{cor}[thm]{Corollary}
\newtheorem{prop}[thm]{Proposition}

\newtheorem{rmk}[thm]{Remark}
\newtheorem*{defi}{Definition}
\theoremstyle{example}
\newtheorem*{example}{Example}

\newcommand{\Hil}{\mathcal{H}}
\newcommand{\slot}{{\,\cdot \,}}

\newcommand{\Spin}{\mathop{\mathsf{Spin}}}
\newcommand{\SO}{\mathop{\mathsf{SO}}}

\newcommand{\PSL}{\mathop{\mathsf{PSL}}}
\newcommand{\SU}{\mathop{\mathsf{SU}}}
\newcommand{\PSU}{\mathop{\mathsf{PSU}}}
\newcommand{\U}{{\mathsf{U}}}

\newcommand{\Mob}{\mathsf{M\ddot ob}}
\newcommand{\Mobc}{\widetilde{\mathsf{M\ddot ob}}}

\newcommand{\Sc}{\mathbb{S}^1}

\renewcommand{\O}{\mathcal{O}}

\newcommand{\cO}{\mathcal{O}}
\newcommand{\cB}{\mathcal{B}}
\newcommand{\cU}{\mathcal{U}}
\newcommand{\cJ}{\mathcal{J}}

\newcommand{\K}{\mathcal{K}}

\newcommand{\cE}{\mathcal{E}}
\newcommand{\cC}{\mathcal{C}}
\newcommand{\cS}{\mathcal{S}}
\newcommand{\cK}{\mathcal{K}}
\newcommand{\A}{\mathcal{A}}
\newcommand{\cI}{\mathcal{I}}
\newcommand{\Aext}{\mathcal{A}_{\mathrm{ext}}}
\newcommand{\M}{\mathcal{M}}

\newcommand{\RR}{\mathbb{R}}
\newcommand{\TT}{\mathbb{T}}
\newcommand{\RRc}{\overline{\mathbb{R}}}
\newcommand{\CC}{\mathbb{C}}

\newcommand{\ZZ}{\mathbb{Z}}
\newcommand{\NN}{\mathbb{N}}

\DeclareMathOperator{\End}{End}
\DeclareMathOperator{\Dom}{Dom}
\DeclareMathOperator{\Hom}{Hom}
\DeclareMathOperator{\Vir}{Vir}
\DeclareMathOperator{\Std}{Std}

\DeclareMathOperator{\Ad}{Ad}

\DeclareMathOperator{\im}{Im}
\DeclareMathOperator{\id}{id}

\DeclareMathOperator{\tr}{tr}
\DeclareMathOperator{\B}{B}
\DeclareMathOperator{\supp}{supp}
\DeclareMathOperator{\sign}{sign}
\DeclareMathOperator{\Aut}{Aut}
\newcommand{\punkt}{\,\mathrm{.}}
\newcommand{\komma}{\,\mathrm{,}}

\newcommand{\e}{\mathrm{e}}
\newcommand{\dd}{\mathrm{d}}
\newcommand{\ima}{\mathrm{i}}

\renewcommand{\Im}{\im}

\newcommand{\Lp}{\mathsf{L}}
\newcommand{\Lpc}{\mathcal{L}}
\newcommand{\Hol}{\mathrm{H\ddot ol}}

\newcommand*\longbijects{%
    \ensuremath{\lhook\joinrel\relbar\joinrel\twoheadrightarrow}
}

\newcommand{\eM}{M}
\newcommand{\emm}{m}

\usepackage{xspace}
\DeclareRobustCommand{\eg}{e.g.\@\xspace}
\DeclareRobustCommand{\cf}{cf.\@\xspace}
\DeclareRobustCommand{\ie}{i.e.\@\xspace}
\makeatletter
\DeclareRobustCommand{\etc}{%
    \@ifnextchar{.}%
        {etc}%
        {etc.\@\xspace}%
}
\makeatother
\newcommand{\BQFTnet}{boundary net\@\xspace}
\newcommand{\BQFTnets}{boundary nets\@\xspace}
\newcommand{\CFTnet}{conformal net\@\xspace}
\newcommand{\CFTnets}{conformal nets\@\xspace}
\newcommand{\Uonenet}{$\U(1)$-current net\@\xspace}
\newcommand{\Cstar}{$C^\ast$\@\xspace}

\newcommand{\three}{III\@\xspace}
\newcommand{\threeone}{III$_1$\@\xspace}

\begin{document}
\date{\today}
\dateposted{\today}
\title
[\MakeUppercase{Models in BQFT Associated with Lattices
and Loop Group Models}]
{\Large \bf Models in Boundary Quantum Field Theory Associated with Lattices
and Loop Group Models}
\author[Marcel Bischoff]{{\sc Marcel Bischoff}\\
Dipartimento di Matematica, Università di Roma ``Tor Vergata''\\
Via della Ricerca Scientifica, 1, I-00133 Roma, Italy\\
E-mail address: {\tt bischoff@mat.uniroma2.it}
}
\thanks{Supported in part by the ERC Advanced Grant 227458
OACFT ``Operator Algebras and Conformal Field Theory''.
}

\begin{abstract}
    In this article we give new examples of models in boundary quantum field 
    theory, i.e. local time-translation covariant nets of von Neumann algebras,     
    using a recent construction of Longo and Witten, which uses a 
    local conformal net $\A$ on the real line together with an element of a 
    unitary semigroup associated with $\A$.
    Namely, we compute elements of this semigroup coming from Hölder continuous
    symmetric inner functions for a family of (completely rational) conformal 
    nets which can be obtained by starting with nets of real subspaces, 
    passing to its second quantization nets and taking local extensions of the 
    former. 
    This family is precisely the family of conformal nets associated with 
    lattices,
    which as we show contains as a special case the level 1 loop group nets of 
    simply connected, simply laced groups.
    Further examples come from the loop group net of $\Spin(n)$ at level 2 
    using the orbifold construction.
\end{abstract}

{
\makeatletter
\let\uppercasenonmath\@gobble
\let\MakeUppercase\relax
\makeatother
\maketitle}

\setcounter{tocdepth}{1}
\tableofcontents

\section{Introduction}\label{sec:Intro}
In the operator algebraic approach to quantum field theory (QFT) one studies 
nets of operator algebras (\eg von Neumann algebras) that assign to a 
space-time region the algebra of observables localized in it. 
These nets are asked to fulfill certain axioms coming from basic physical 
principles; we mention as examples the locality principle---which asks that 
the algebras assigned to causally disjoint regions should commute 
(local nets)---and the covariant assignment with respect to some 
``symmetry group'' of the space-time 
(for a general introduction on this subject we refer to the textbook \cite{Ha}).

In this approach also conformal quantum field theory (CQFT) 
has been treated by considering nets on two dimensional Minkowski space
and its chiral parts, which can be regarded as nets on the real line 
or as nets on the circle.
Besides CQFT on the full Minkowski space 
also boundary conformal quantum field theory (BCFT) on
the Minkowski half-plane $x>0$ is described in the algebraic approach.
More precisly, in the paper \cite{LoRe2004} Longo and Rehren associate with a 
local conformal net $\A$ 
on the real line a local conformal boundary net $\A_+$ on the Minkowski half-plane
and obtain more general boundary nets which are extending $\A_+$.

Lately, in \cite{LoWi2010} Longo and Witten have given a framework to construct 
models in boundary quantum field theory (BQFT) by investigating into local nets
on the Minkowski half-plane, which are in general only time-translation 
covariant and can be considered as a deformation of the net $\A_+$.
Specifically, the construction starts with a conformal net $\A$ on the real 
line together with an element $V$ of a unitary semigroup $\cE(\A)$ 
associated with $\A$  to construct a net on Minkowski half-plane, where the 
special case $V=1$ is the net $\A_+$. 
The search for new models is basically given by the construction of elements of 
the semigroup $\cE(\A)$ for a given conformal net $\A$.
Further in \cite{LoRe2011} Longo and Rehren investigate in BQFT on the Lorentz
hyperboloid using a similar semigroup, nonetheless we will concentrate in this 
paper on BQFT on Minkowski half-plane.
In this framework, an interesting class of conformal nets to 
consider are the completely rational 
conformal nets \cite{KaLoMg2001} just having finite number of sectors
(equivalence classes of irreducible representations) with each one having only
finite statistics and their representation theory giving rise to modular tensor
categories.

In this work the main goal is to construct elements of the semigroup 
$\cE(\A)$ for the loop group models (which are ``in general'' expected to 
fulfill complete rationality) and hence give new models of BQFT.
The loop group models are conformal nets coming from (projective) 
positive energy representations of loop groups.
The cocycles of the projective positive energy representation of loop 
groups are classified by their level 
and the vacuum representation of each level yields a conformal net. 
In the case of simply laced Lie groups the level 1 representation is the basic 
representation and all higher level representations are contained in tensor 
products \cite{GaFr1993}; as a result the higher level loop group models are
contained as subnets in the tensor product of the level 1 loop group net. 
So the first important step is to construct semigroup elements for the
level 1 loop group net. 
This can be obtained as a subnet of a free Fermionic net or as an extension of 
a free Bosonic net (as we will show here); by free nets we mean second
quantization nets using the CCR or CAR algebra of a net of real subspaces 
in the Bose and Fermi case, respectively.
For these free nets the semigroup elements being 
second quantization unitaries are characterized in  \cite{LoWi2010}.

We look into extensions of free Bosons, namely the family of \CFTnets
associated with lattices and show that this family indeed contains the level 1 
loop group models of simply laced groups as a special case. 
This can be regarded as an algebraic version of the Frenkel--Kac or 
Frenkel--Kac--Segal construction, which says that the lattice 
vertex operator algebras of 
simply laced root lattices correspond to the level 1 Kac--Moody vertex 
operator algebras (cf. Theorem 5.6 \cite{Ka1998}).
Furthermore this family consists only of completely ration conformal nets 
as shown by Dong and Xu in \cite{DoXu2006}.

In Section \ref{sec:standard} we give some basic preliminaries on 
standard subspaces, 
its associated modular theory and semigroups of standard pairs. 

In Section \ref{sec:Ch} we review the construction of the conformal nets under
investigation starting with a net of standard subspaces in 
the spirit of \cite{BGL2002,Lo} and obtaining a net of von Neumann algebras by
second quantization, which describes 
        Abelian currents on the circle. 
Their local extensions by even lattices are shown 
to be the conformal nets associated with lattices constructed in
\cite{DoXu2006} using positive energy representations of 
the loop group of the torus related to the lattice. 
We show that this family indeed contains
as a special case the loop group nets at level 1 of simply laced groups.

In Section \ref{sec:BQFT} we bring about a family of local nets on the Minkowski 
half-plane associated with each step of the former construction of chiral 
models.
The important step is the extension of the semigroup elements 
obtained by second quantization to the local extension by a lattice. We give 
criterion when such elements extend and also 
look into restriction to subnets. A further family of examples
of semigroup elements and therefore models in BQFT are calculated for 
the loop group nets of $\Spin(n)$ at level 2 using the orbifold construction.
\section{Preliminaries on Standard Subspaces}\label{sec:standard}
In this section we give some basic preliminaries on 
standard subspaces, 
its associated modular theory and semigroups of standard pairs.
\subsection{Standard Subspaces} 
We repeat some basic facts (for details see \cite{Lo}) on standard 
subspaces.
Let $\Hil\equiv (\Hil,(\slot,\slot))$ be a Hilbert space and let $H\subset \Hil$
be a real subspace. 
We denote by $H'=\{x\in \Hil: \im(x,H)=0\}$ the 
    \emph{symplectic complement}, 
which is closed. In particular it is $H'' = \overline H$.
A closed real subspace $H$ is called 
    \emph{cyclic} 
if $\overline {H + \ima H} = \Hil$ and is called 
    \emph{separating} 
if $H \cap \ima H = \{0\}$. 
So a closed real subspace $H$ is separating or cyclic if and only if its
symplectic complement $H'$ is cyclic or separating, respectively. 
A cyclic and separating subspace $H$ is called \emph{standard}; clearly $H$ is 
standard if and only if $H'$ is standard.
We denote the set of all standard subspaces of $\Hil$ by $\Std(\Hil)$.
To a standard subspace we relate a pair $(J_H,\Delta_H)$, where 
$(\Delta_H^{\ima t})_{t\in\RR}$ is a unitary one-parameter group called the
    \emph{modular unitaries} 
and an antiunitary involution $J_H$ called 
    \emph{modular conjugation}. 
Both are defined by the polar decomposition of the densly defined, closed, 
antilinear involutive (i.e. $S_H^2 \subset \id_\Hil$) operator 
$S_H=J\Delta_H^{1/2}$ with domain $H+\ima H$ defined by
$x+\ima y \longmapsto x-\ima y$ for $x,y \in H$. 
A (simplier) real subspace version of the Tomita-Takesaki theorem gives:
\begin{align*}   
    JH&=H', \qquad \Delta^{\ima t}H = H  &(t\in\RR).
\end{align*} 
We note that there is a useful bijective correspondence between $\Std(\Hil)$
and the set of densely defined, closed, antilinear involutions $S$ on $\Hil$, 
given by the map $H \longmapsto S_H$ as above, with inverse map associating 
with such an involution $S$ the standard subspace 
$H_S = \{ x\in \Dom(S) : Sx=x\}=\ker(1-S)$.

\subsection{Semigroup Associated with Standard Pairs}\label{seq:BQFTStdSemigroup}
\begin{defi} Let $H$ be a standard subspace of a Hilbert space 
    $\Hil$ and let us assume that there exists a one-parameter group 
    $T(t)=\e^{\ima t P}$ on $\Hil$ such that:
    \begin{itemize}
        \item$T(t) H \subset H$ for all $t\geq 0$,
        \item$P>0$.
    \end{itemize}
    Then we call the pair $(H,T)$ a \emph{standard pair}. It is called 
    \emph{non-degenerated} if the kernel of $P$ is $\{0\}$.
\end{defi}
A one-particle version of Borchers Theorem with some implications holds:
\begin{thm}[\cite{LoWi2010}*{Theorem 2.2}] Let $(H,T)$ be a non-degenerate standard pair. 
    \begin{enumerate}
    \item Then for all $t,s\in \RR$ holds:
    \begin{align*}
        \Delta^{\ima s} T(t) \Delta^{-\ima s} &= T(\e^{2\pi s} t),
        &
        JT(t)J &=T(-t),
    \end{align*}
    where $\Delta^{\ima t}$ and $J$ are the modular unitaries and conjugation, respectively, associated with the standard
    space $H$, i.e. $JH=H'$ and $\Delta^{\ima t} H = H$.
    \item 
    $(H,T)$ yields a unitary positive energy representation of the 
    translation-dilation group of $\RR$ also called 
    the $ax+b$ group, by 
    associating with $x\longmapsto \e^{-2\pi s}x +t$ the unitary element 
    $T(t)\Delta^{\ima s}$.
    \item There is a unique irreducible standard pair and each standard pair is a multiple of this unique standard pair.
\end{enumerate}
\end{thm}
\begin{defi} Let $(H,T)$ be a standard pair on $\Hil$. 
    The semigroup of unitaries $V$ of $\Hil$ commuting with $T$ such that 
    $VH\subset H$ is denoted by $\cE(H,T) = \cE(H)$.
\end{defi}
The elements of $\cE(H)$ are characterized in \cite{LoWi2010}. 
We first state the case of the irreducible standard pair,
where the semigroup $\cE(H_0)$ can be identified with a semigroup 
of certain ``symmetric inner functions''. 
\begin{defi}
    We denote by $\cS$ the set of all complex Borel functions 
    $\varphi:\RR \longrightarrow \CC$ which are boundary 
    values of a bounded analytic function 
    on $\RR+\ima \RR_+$, which are \emph{symmetric}, i.e.
    $\overline{\varphi(p)} = \varphi(-p)$ and \emph{inner}, i.e. 
    $|\varphi(p)|=1$ for almost all $p\geq 0$. 
\end{defi}
\begin{thm}[{\cite[Corollary 2.4]{LoWi2010}}]\label{thm:stdir} 
     Let $(H_0,T_0)$ be the unique irreducible standard pair then $V\in\cE(H)$ 
     if and only if $V=\varphi(P)$  for some $\varphi\in\cS$. 
\end{thm}
In the reducible case the semigroup $\cE(H)$ consists of 
matrices of similar functions and the 
condition $|f(p)|=1$ is generalized to unitarity of the matrix. 
\begin{rmk} \label{rmk:HT}
    Let $(H,T)$ be a non-zero, non-degenerated standard pair on a Hilbert space
    $\Hil$. Then it can be decomposed as a direct sum of the unique irreducible
    standard pair. Let 
    \begin{align*} 
        \Hil &= \bigoplus_i \Hil_i & 
        H &= \bigoplus_i H_i& 
        T &= \bigoplus_i T_i
    \end{align*}
    be such a finite or infite decomposition, where each $(H_i,T_i)$ is a 
    standard pair in $\Hil_i$ and can be identified with the unique irreducible 
    standard pair $(H_0,T_0)$ with generator $P_0$.
\end{rmk}
\begin{defi}
    For $n \in \NN \cup \{\infty\}$ we denote by $\cS^{(n)}$ the set of 
    matrices $(\varphi_{hk})_{1\leq h,k\leq n}$ where 
    $\varphi_{hk}:\RR \longrightarrow \CC$ are complex Borel functions which 
    are boundary values of a bounded analytic function on $\RR+\ima\RR_+$
    such that $\varphi_{hk}(p)$ is a unitary matrix for almost all $p$,
    which is symmetric, i.e. $\overline{\varphi_{hk}(p)}=\varphi_{hk}(-p)$.
\end{defi}
\begin{thm}[{\cite[Theorem 2.6]{LoWi2010}}] \label{thm:stdred} 
    Let $H$ be like in Remark \ref{rmk:HT}. Then $V\in \cE(H)$ if and only if 
    it is a $n \times n$ matrix $(V_{hk})$
    with entries in $\B(\Hil)$ such that 
    $V_{hk}=\varphi_{hk}(P_0)$ for some $(\varphi_{hk})\in\cS^{(n)}$.    
\end{thm}

\section{Conformal Field Theory -- Conformal Nets}\label{sec:Ch}
In this section we are interested in local Möbius covariant nets (\CFTnets). 
These are nets on the circle (or the real line), which physically describe 
the chiral part of the algebra of observables of a 2D  
QFT, where the real line (circle) 
is then identified with (the compactification) of one of the lightrays.

\subsection{Nets of Standard Subspaces}\label{sec:ChStdNet}
Before describing nets of von Neumann algebras we want to go a step back and 
give some details on nets of real subspaces of a Hilbert space $\Hil_0$, whose
``second quantization'' leads to nets of von Neumann algebras, the so called 
second quantization nets.
In analogy to the ``free field construction'' from Wigner particles the 
Hilbert space $\Hil_0$ will be called the ``one-particle space''. 
See for example \cite{BGL2002} for a general construction of free Bosons 
using this technique on more general space-times%
\footnote{In our case the  ``space-time'' is the circle and the 
``wedges'' correspond to open 
non-empty nowhere dense intervals} and \cite{Lo} for such nets
on the circle.

We will  identify the one-point compactification 
$\RRc = \RR \cup\{\infty\}$ of the real line with the circle 
$\Sc = \lbrace z \in \CC : |z|=1 \rbrace$ by the Caley map
\begin{align*}
    C: \RRc &\longbijects \Sc,~
     x\longmapsto   -\frac{x-\ima}{x+\ima} 
     \quad \Longleftrightarrow\quad
     x=C^{-1}(z) = -\ima \frac{z-1}{z+1}. 
\end{align*}
Our symmetry is the 
    \emph{group of Möbius transformations}
$\Mob$ of the circle and can be identified with either $\PSL(2,\RR)$ or 
$\PSU(1,1)$, which act naturally on the compactified real line $\RRc$ and the 
circle $\Sc$, respectively. 
The Möbius group is generated by the following three one-parameter subgroups: 
the 
    \emph{rotations} 
$R(\theta)z = \e^{\ima\theta}z$, which are easier in the circle picture; 
the
    \emph{translations}
$\tau(t)x=x+t$ and 
    \emph{dilations} 
$\delta(s)x=\e^s x$ for $x\in \RRc$ which are both easier in the real line 
picture.
We add the orientation reversing 
    \emph{reflection}
$rz=\bar z$ with $z\in\Sc$ to $\Mob$ and denote the obtained group by 
$\Mob_\pm=\Mob\rtimes_{\Ad r}\ZZ_2$.
For $z\in\Sc$ we sometimes write $z=\e^{\ima\theta}$ and note that 
$x \equiv C^{-1}(\e^{\ima \theta})=\tan (\theta/2)$.

A connected, non-empty, nowhere dense interval $I\subset\Sc$ is called 
    \emph{proper}
and we denote by $\cI$ the set of all proper 
intervals partially ordered by inclusions. 
For $I\in\cI$ we denote by $I'$ the interior of $\Sc\setminus I$ which is 
clearly in $\cI$ and note that $\Mob$ acts transitive on $\cI$.

\begin{defi} 
    A strongly continuous unitary representation of $\Mob$ (or a subgroup
    containing the rotations) on a Hilbert 
    space $\Hil$ is called \emph{positive energy representation} if the 
    generator $L_0$ of the one-parameter subgroup of rotations 
    $U(R(\theta))=\e^{\ima \theta L_0}$ has positive spectrum.
    The representation is called 
        \emph{non-degenerate}
    if it does not contain the trivial representation.
\end{defi}
\begin{rmk}[{\cite[Theorem 2.10]{Lo}}] A unitary positive energy 
    representation of $\Mob$ extends to a (anti-) unitary representation of 
    $\Mob_\pm$ on the same Hilbert space and the extension is unique
    up to unitary equivalence.
\end{rmk}
\begin{defi}
    A local Möbius covariant net of standard subspaces of $\Hil$ is a 
    family of standard subspaces $H(I) \subset \Hil$ indexed by $I\in\cI$ such 
    that the following properties hold:
    \begin{enumerate}[{\bf A.}]
        \item\textbf{Isotony.} $I_1 \subset I_2$ implies 
            $H(I_1) \subset H(I_2)$. 
        \item\textbf{Locality.} $I_1\cap I_2=\emptyset$ implies
            $H(I_1) \subset H(I_2)'$.
        \item\textbf{Möbius covariance.} 
            There is a positive energy representation of $\Mob$ on $\Hil$ 
            such that $U(g)H(I) = H(gI)$ for all $g\in \Mob$ and $I \in \cI$.
        \item\textbf{Irreducibility.} $U$ is non-degenerate, i.e. 
            does not contain the trivial representation.
    \end{enumerate}
\end{defi}

Given a positive energy representation $U$ of $\Mob$ on $\Hil$ 
we can construct a local Möbius covariant net of standard subspaces 
as follows: we define the unitary one-parameter group 
$\Delta^{\ima t} = U(\delta(-2\pi t))$ where 
$\delta(t)x = \e^t x$ are the 
dilations and the antiunitary involution $J=U(r)$ 
(where we use that $U$ extends to a representation
of $\Mob_\pm$) and define the densely defined, closed, antilinear involution
$S=J\Delta^{1/2}$.
We denote by $I_0$ the interval corresponding to the upper circle or 
equivalently $(0,\infty)$. Then we set 
$H(I_0)\equiv H(0,\infty) = \{x\in\Dom(S): Sx=x\}$
to be the standard subspace associated with $S$ and for general 
$\cI\ni I=gI_0$ we set $H(I)=U(g)H(0,\infty)$, which does not 
depend on the choice of $g\in\Mob$. All local Möbius covariant nets
of standard subspaces are obtained in this way \cite{Lo}.

For later use we make the construction of a family indexed by $n\in \NN$ of 
local Möbius covariant nets of real subspaces---namely the net coming from $n$ 
copies of the lowest weight 1 positive energy representation (cf. \cite{Lo}) 
of $\Mob$---more explicit. Therefore let $F$ be a $n$-dimensional 
Euclidean space with scalar product 
$\langle \slot, \slot\rangle$. 
Let us define 
$\Hil_{0,F} = \Hil_0 \otimes_\RR F\cong \bigoplus_{i=1}^n \Hil_0$ which is in 
particular isomorphic to the $n$-fold direct sum of the unique irreducible 
positive energy lowest weight representation of $\Mob$ with lowest weight 1 
denoted by $(U_0,\Hil_0)$. 
We denote by $U_{0,F}$ the unitary representation of the $\Mob$ on 
$\Hil_{0,F}$. It can explicitly be constructed as follows. Let
$\Lp F=C^\infty(\Sc,F)\cong C^\infty(\Sc,\RR)\otimes_\RR F$ the set of all 
smooth maps (loops) from the circle $\Sc$ to $F$.
Because $f\in \Lp F$ is periodic it can be written as a Fourier series
\begin{equation*}
    f(\theta) = \sum_{k\in\ZZ} \hat f_k \e^{\ima k \theta}, \qquad
    \hat f_k = \int_0^{2\pi} \e^{-\ima k\theta}f(\theta) \frac{\dd \theta}{2\pi}
\end{equation*} 
with Fourier coefficients $\hat f_k = \overline{\hat f_{-k}}$ in the 
complexified space  $F_\CC:=F\otimes_\RR \CC$.
We introduce a semi-norm
\begin{equation*}
    \| f \|^2= \sum_{k=1}^\infty k\cdot \| \hat f_k\|^2_{F_\CC}
\end{equation*}
and a complex structure, i.e. an isometry $\cJ$ w.r.t. $\|\slot\|$ satisfying 
$\cJ^2=-1$, by
\begin{align*}
    \cJ:\hat f_k            &\longmapsto -\ima \sign(k) \hat f_k
\end{align*}
and finally we get the Hilbert space $\Hil_{F,0}$ by completion with respect to the norm $\|\slot\|$ 
\begin{align*} 
    \Hil_{0,F} &= \overline{\Lp F /F}^{\|\slot\|}
    \komma
\end{align*}
where $F$ is identified with the constant functions.
The scalar product 
$(\slot,\slot)$ can be obtained by polarization and the unitary action of 
$\Mob$ is induced by the action on $\Lp F$, namely
 \begin{align*} 
    U(g):f&\longmapsto g_\ast f:(g_\ast f)(\theta)=f(g^{-1}(\theta))
    \punkt
\end{align*}
Let $f\in \Lp F$. If no confusion arises we denote also its image 
$[f]\in\Hil_{0,F}$
of the inclusion $\iota_F:\Lp F \rightarrow \Hil_{0,F}$ 
by $f$.
On $\Lp F$ the sesquilinear form coming from the scalar product is given 
explicitly by
\begin{equation*} 
    \omega(f,g) := \Im(f,g)=\frac{-\ima}2\sum_{k\in\ZZ}k \langle \hat f_{k}, \hat g_{-k} \rangle 
    = \frac 1 2 
    \int_0^{2\pi}\langle f(\theta),g'(\theta)\rangle\frac{\dd\theta}{2\pi} =:  \frac12\int \langle f,g'\rangle
    \punkt
\end{equation*}

For $I\in \cI$ we denote by $H_F(I)$ the closure subspace of functions with 
support in $I$. The family $\{H_F(I)\}_{I\in\cI}$ is a local Möbius covariant 
net of standard subspaces. 
Indeed because $U$ acts geometrical, and in particular $U(\delta(t))$ is 
the modular group of the abstract construction and leaves $H_F(0,\infty)$ 
invariant, one can show that the explicit construction equals the modular 
construction mentioned above (cf. \cite{Lo}). 
\begin{prop}\label{prop:NetStd} Let $(F,\langle\slot,\slot\rangle)$ be 
    an Euclidean space, then there is a local Möbius covariant net 
    of standard subspace $H_F$ on the Hilbert space $\Hil_{0,F}$.
\end{prop}
We remark that by the geometric modular action follows that the net is 
\emph{Haag dual}, \ie $H_F(I')=H_F(I)'$ and also the restriction to $\RR$ 
can be shown to be  Haag dual,
\ie $H_F( (\RR\setminus I)^\circ)= H_F(I)$ for $I \Subset \RR $.

\subsection{Conformal Nets}
In this part we give the notion of a \emph{local Möbius covariant net} 
of von Neumann algebras which we will simply call \emph{\CFTnet}.
\begin{defi}
    A \emph{local Möbius covariant net (\CFTnet)} $\A$ on $S^1$ 
    is a family $\{\A(I) \}_{I\in \cI}$ of von Neumann algebras on a Hilbert 
    space $\Hil$, with the following properties:
    \begin{enumerate}[{\bf A.}]
        \item \textbf{Isotony.} $I_1\subset I_2$ implies 
            $\A(I_1)\subset \A(I_2)$.
        \item \textbf{Locality.} $I_1  \cap I_2 = \emptyset$ implies 
            $[\A(I_1),\A(I_2)]=\{0\}$.
        \item \textbf{Möbius covariance.} There is a unitary representation
            $U$ of $\Mob$ on $\Hil$ such that 
            $  U(g)\A(I)U(g)^\ast = \A(gI)$.
        \item \textbf{Positivity of energy.} $U$ is a positive energy 
            representation, i.e. the generator $L_0$ (conformal Hamiltonian) 
            of the rotation subgroup $U(R(\theta))=\e^{\ima \theta L_0}$
            has positive spectrum.
        \item \textbf{Vacuum.} There is a (up to phase) unique rotation 
            invariant unit vector $\Omega \in \Hil$ which is 
            cyclic for the von Neumann algebra $\bigvee_{I\in\cI} \A(I)$.
    \end{enumerate}
\end{defi}
The 
    \emph{Reeh--Schlieder property} 
holds automatically
\cite{FrJr1996}, i.e. $\Omega$ is cyclic and separating for any $\A(I)$ with 
$I\in\cI$. 
Further we have the 
    \emph{Bisognano--Wichmannn property} 
\cite{GaFr1993,BrGuLo1993} which states that the modular operators with respect to 
$\Omega$ have geometric meaning; \eg the modular operators for the 
upper circle $I_0$ are given by the dilation 
$\Delta^{\ima t}=U(\delta(-2\pi t))$ and reflection $J=U(r)$, where here $U$ is 
extended to $\Mob_\pm$. 
For a general interval $I\in\cI$ the modular operators are given by a special
conformal transformation $\delta_I$ and a reflection $r_I$ both fixing the 
endpoints of $I$. 
The Bisognano--Wichmannn property implies \emph{Haag duality} 
\begin{align*}
    \A(I)'&=\A(I') &I\in\cI
\end{align*}
and it can be shown (see \eg \cite{GaFr1993}) that each $\A(I)$ is a type 
\threeone factor in Connes classification \cite{Co1973}.
A conformal net is 
    \emph{additive} 
\cite{FrJr1996}, \ie for intervals $I,I_1,\ldots I_n \in \cI$
$$
    I \subset \bigcup_i I_i \quad\Longrightarrow\quad 
    \A(I) \subset \bigvee_{i} \A(I_i)
    \text{ holds.}
$$

\subsubsection{Representations} 
Let $\A$ be a \CFTnet on a Hilbert space $\Hil$. 
A \emph{covariant representation} $\pi=\{\pi_I\}_{I\in\cI}$ is a family of 
representations $\pi_I$ of $\A(I)$ on a fixed Hilbert space $\Hil_\pi$ which 
fulfill:
\begin{align*}
    \pi_{I}\restriction_{\A(I_0)} &= \pi_{I_0} & I_0\subset I
\\
    \Ad U_\pi(g) \circ \pi_I &= \pi_{gI} \circ \Ad U(g)
\end{align*}
where $U_\pi$ is a unitary representation of the universal covering group 
$\Mobc$ of $\Mob$ with positive energy.
We assume $\Hil_\pi$ to be separable and this implies that $\pi$ is 
    \emph{locally normal}, 
namely $\pi_I$ is normal for all $I\in\cI$. 
A representation $\rho$ is called 
    \emph{localized} 
in some interval $I_0\in\cI$ if $\Hil_\rho=\Hil$ and 
$\rho_{I_0'}=\id_{\A(I_0')}$. 
Due to the type \threeone{} factor property, each representation $\pi$ is 
localizable in any interval $I_0\in\cI$, namely there is a representation 
$\rho$ which is unitary equivalent to $\pi$ and localized in $I_0$. 
If $\rho$ is localized in $I_0\in\cI$ then by Haag duality for every 
$I\in \cI$ with $I \supset I_0$ it is $\rho_I(\A(I))\subset  \A(I)$, in other
words $\rho_I$ is an endomorphism of $\A(I)$.
Let $\rho$ be a (covariant) representation localized in $I_0$. 
By a \emph{local cocycle} \cite{Lo2003} localized in a proper interval 
$I \supset I_0$, we mean an assignment of a symmetric neighbourhood 
$\cU$ of the identity of $\Mobc$ such that $I_0\cup gI_0\subset I$ for all 
$g\in\cU$  and a strongly continuous unitary valued map 
$g\in\cU \longmapsto z_\rho(g) \in \A(I)$ such that with 
$\alpha_g:=\Ad U(g)$:
\begin{align*}
    z_\rho(gh)&=z_\rho(g)\alpha_g(z_\rho(h))
    & g,h\in\cU
    \\
    \Ad z_\rho(g)^\ast \circ \rho_{\tilde I}(a)&=
    \alpha_g\circ\rho_{g^{-1}\tilde I}\circ \alpha_{g^{-1}}(a) &
    g\in \cU, a\in\A(\tilde I)
\end{align*}
for some open interval $\tilde I \in \cI$ with $\tilde I\supset \overline I$.
By covariance and Haag duality there exists a local cocycle
given by
\begin{align*}
    z_\rho(g)&=U_\rho(g)U(g)^\ast \in \A(I) &g\in\cU
    \punkt
\end{align*}

\subsubsection{Conformal subnets} Let $\A$ be a conformal net and $U$ its 
associated positive energy representation of $\Mob$. 
We call a family $\{\cB(I)\}_{I\in\cI}$ with $\cB(I)\subset \A(I)$ for all 
$I\in\cI$ a \emph{conformal subnet} if $\cB$ is isotonous, \ie $I,J\in\cI$ 
with $I\subset J$ implies $\cB(I)\subset\cB(J)$ and covariant, \ie it is 
$U(g)\cB(I)U(g)^\ast =\cB(gI)$ for all $I\in\cI$ and $g\in\Mob$.
The structure of conformal subnets is studied in \cite{Lo2003}.

Let $e$ be the projection on the closure $\Hil_\cB$ of $\bigvee_{I\in\cI}\cB(I)\Omega$. 
Then $\cB$ is itself a conformal net on $\Hil_\cB:= e\Hil$ with unitary 
representation $U\restriction_{\Hil_\cB}$ also denoted by $U$, 
namely $\Omega$ is cyclic for 
$\Hil_\cB$ by definition and all other properties are inherit by the ones of 
$\A$.
By the Reeh--Schlieder property $\Omega$ is cyclic and separating for 
all $\cB(I)$ with $I\in \cI$ and in particular $e$ is the Jones projection 
(see e.g. \cite{LoRe1995}) of the inclusion $\cB(I) \subset 
\A(I)$.
\begin{lem}\label{lem:subnet} Let $\cB$ be a conformal subnet of $\A$.
    If $e =1$ then the conformal nets $\cB$ and $\A$ are identical.
\end{lem}
\begin{proof}
    Let $I\in\cI$. Then it is $\cB(I)\subset \A(I)$ and by the 
    Bisognano--Wichmann property the modular group of $\A(I)$ with respect to 
    the vector state of $\Omega$ is given by 
    $\sigma_t = \Ad U(\delta_I(-2\pi t))$ 
    and by covariance of $\cB$ it leaves $\cB(I)$ invariant. 
    By Takesaki's Theorem \cite[Theorem IX.4.2.]{Ta2} there exists a normal conditional 
    expectation 
    from $\A(I)$ onto $\cB(I)$ which has to be the identity of $\A$ by $e=1$.
\end{proof}

\subsubsection{Completely rational conformal nets}
A conformal net $\A$ is said to be 
    \emph{strongly additive} 
if for $I_1,I_2 \in \cI$ adjacent intervals and $I=(I_1\cup I_2)''=
\overline{I_1\cup I_2}^\circ \in \cI$,
$$
    \A(I_1) \vee \A(I_2) =\A(I)\text{ holds.}
$$
The net $\A$ is called 
    \emph{split} 
if for $I_0,I\in \cI$ with $\overline{I_0}\subset I$ the inclusion 
$\A(I_0) \subset \A(I)$ is a split inclusion, namely there exist an 
intermediate type I factor $M$ such that $\A(I_0) \subset M \subset \A(I)$ 
or equivalently $\A(I_0) \vee \A(I)'$ is canonically isomorphic to 
$\A(I_0)\otimes \A(I)'$.
Let $I_1,I_3 \in \cI$ be two intervals with disjoint closure and 
$I_2,I_4\in\cI$  the two components of $(I_1\cup I_3)'$, in other words the 
intervals $I_1,\ldots,I_4$ divide the circle into four parts. 
Then we denote by $\mu_\A$ the Jones--Kosaki index \cite{Ko1998} of the 
inclusion 
\begin{equation}
    \A(I_1)\vee \A(I_3) \subset (\A(I_2) \vee \A(I_4))'
    \label{eq:twointerval}
\end{equation}
which does not depend on the special choice of the intervals $I_i$.
Finally the net $\A$ is called 
    \emph{completely rational} 
if it is strongly additive, split and $\mu_\A <\infty$. 
In \cite{KaLoMg2001} it is shown 
that the index of the inclusion \eqref{eq:twointerval}
is the global index associated with all sectors 
and the the category of representations form a modular 
tensor category, where each sector is 
a direct sum of sectors with finite dimension.

\subsection{Second Quantization Nets}
\label{sec:ChAbNet}

By second quantization of a net of standard subspaces we become 
a net of von Neumann algebras.

Let $\Hil$ be a Hilbert space and $\omega(\slot,\slot)=\Im(\slot,\slot)$ the
sesquilinear form.
There are unitaries $W(f)$ for $f\in\Hil$ fulfilling
\begin{equation*}
    W(f)W(g) = \e^{-\ima \omega(f,g)} W(f+g) = \e^{-2\ima\omega(f,g)}W(g)W(f).
\end{equation*}
and acting naturally on the Bosonic Fock space $\e^\Hil$ over $\Hil$.
This space is given by $\e^{\Hil} = \bigoplus_{n=0}^\infty P_n\Hil^{\otimes n}$,
where $P_n$ is the projection 
$P_n(x_1\otimes\cdots\otimes x_n)=
1/n!\sum_{\sigma} x_{\sigma(1)}\otimes\cdots\otimes x_{\sigma(n)}$ where the 
sum goes over all permutation. 
The set of coherent vectors
$\e^h := \bigoplus_{n=0}^\infty h^{\otimes n}/\sqrt{n!}$ with $h\in\Hil$ is 
total in $\e^\Hil$ and it is $(\e^f, \e^h)=\e^{(f,h)}$.
The vacuum is given by $\Omega =\e^0$ and the action of $W(f)$ is given by 
$W(f)e^0=\e^{-\frac12\|f\|^2}\e^{f}$, in other words the vacuum 
representation
$\phi(\slot) = (\Omega,\slot\Omega)$ is characterized by 
$\phi(W(f))=\e^{-\frac 12 \|f\|^2}$.

For a real subspace $H\subset \Hil$ we define the von Neumann algebra 
\begin{align*}
    R(H) = \lbrace W(f):f\in H\rbrace''\subset \mathrm{B}(\e^\Hil)
    \punkt
\end{align*} 
The map $R$ has the following properties:
\begin{prop}[\cite{Lo2}]\label{prop:CCR}{$\quad$}
        \begin{enumerate}
        \item Let $H,K\subset \Hil$ be real linear subspaces. Then
            $R(K) = R(H)$ iff $\bar K = \bar H$.
        \item Let $H$ be closed. $H$ is separating or cyclic
            iff $R(H)$ is separating or cyclic, respectively.
        \item Let $H$ be standard, then the modular unitaries 
            $\Delta_{R(H)}^{\ima t}$ and the modular conjugation $J_{R(H)}$ 
            associated with $(R(H),\Omega)$ are given by 
            \begin{equation*}
                \Delta^{\ima t}_{R(H)} = \Gamma(\Delta_H^{\ima t}), \qquad
                J_{R(H)} = \Gamma(J_H)
            \end{equation*}
            and in particular $R(H') = R(H)'$.
    \end{enumerate}
\end{prop}
Let $U$ be a unitary in $\B(\Hil)$ then 
$\Gamma(U)=\bigoplus_{n=0}^\infty U^{\otimes n}$ acts on coherent states 
by $\Gamma(U)\e^h = \e^{Uh}$ and is therefore a unitary (cf. \cite{Gu2011}) on 
$\e^{\Hil}$.
These second quantization unitaries implement Boguliubov automorphisms, namely 
$\Gamma(U) W(f)\Gamma(U)^\ast = W(Uf)$.
\begin{prop}[Second quantization nets \cite{Lo2}] Let $\{H(I)\}_{I\in\cI}$ be  a local Möbius covariant
    net of standard subspaces on $\Hil$.
    Then $\A(I)=R(H(I))$ defines a local Möbius covariant net 
    (of von Neumann algebras) on $\e^{\Hil}$.
\end{prop}
    
Let $(F,\langle\slot,\slot\rangle)$ be an Euclidean space and
$\cI\ni I\longmapsto H_F(I)\subset \Hil_{0,F}$ the net of standard subspaces 
from Proposition \ref{prop:NetStd}.
Then  we denote by $\A_F$ the local Möbius covariant net 
on $\Hil_F:=\e^{\Hil_{0,F}}$
called the \emph{Abelian current net over $F$} 
given by $\A_F(I):=R(H_F(I))$. 
If $U_0$ is the action of $\Mob$ on $\Hil_{0,F}$ then the action 
on $\Hil_F$ is given by $U(g):=\Gamma(U_0(g))$. 
In the case $F=\RR$ the net is also called the \emph{\Uonenet} 
and was treated in an operator algebraic setting first in \cite{BuMaTo1988}.
We remark that $\A_F$ is clearly equivalent to the $n$-fold tensor product 
of the \Uonenet.

\subsubsection{Representations}
Let $\ell \in C^\infty(S_1,F)$ with support in some $I_0\in \cI$. Then we 
define for $I\in\cI$ with $I_0\subset I$ 
$$
    \rho_{\ell,I}(W(f))=\e^{\ima\int\langle f,\ell\rangle} W(f)
$$
where we have chosen a representant $f$ of $[f]\in \Hil_{F,0}$ with 
$f\restriction_{I'}\equiv 0$. This defines a representation localized 
in $I_0$. 
This representation is covariant with local cocycle
localized in $I\supset I_0$ and $\cU$
a symmetric neighbourhood of the identity of $\Mobc$
such that $I_0\cup gI_0\subset I$ for all $g\in\cU$ given by
$z(g)=W(L-L_g)$ where $L$ is a primitive of $\ell$, i.e.
$L'(\theta)=\ell(\theta)$ and $L_g(\theta)=g_\ast L(\theta)=L(g^{-1}\theta)$.

Two representations are equivalent if they have the same charge,
which is given for $\rho_\ell$ by
$$
    q_\ell=\int_0^{2\pi} \ell(\theta) \frac{\dd\theta}{2\pi}= \int\ell \in F,
$$
namely for $q_\ell=q_\emm$ it is 
$z\rho_\ell = \rho_\emm z$ with unitary intertwiner $z=W(M-L)$ where 
$M-L\in \Hil_{F,0}$ is a primitive of $\emm - \ell$. 
In other words the sectors depend only on this charge $q\in F$ and 
we denote the sector by 
$[q]$ with obvious fusion rules $[q]\times [r]=[q+r]$.
We note that because there are infinitely many sectors
(with dimension 1) the index of the inclusion \eqref{eq:twointerval}
is infinite and the nets cannot be completely rational.

Equivalently the conformal net can be regarded as coming from a projective 
positive energy representation of the group $\Lp F$.

\subsubsection{Abelian currents from central extensions}
Basically to fix notation, we recall
some facts about projective representations. 
If $\pi$ is a projective representation of a group $G$ on a Hilbert space 
$\Hil$, then there is a 2-cocycle with 
$c:G\times G \longrightarrow \TT \subset \CC$ given by  
\begin{align*}
    \pi(g)\pi(h) &= c(g,h)\pi(gh) & \text{for all } g,h \in G
\end{align*}
fulfilling the cocycle relation
    $c(h,k)c(g,hk) = c(g,h)c(gh,k),$
which follows from associativity. 
Two representations are equivalent if and only if there is 
a coboundary 
\begin{equation}
    \label{eq:cob}
    b_f(g,h)=\frac{f(g)f(h)}{f(gh)}
\end{equation}
where $f:G\longmapsto \TT$ such that  
\begin{equation*}
    c_2(g,h) = b_f(g,h)c_1(g,h) \quad \Longleftrightarrow \quad
    \pi_2(g) = f(g)\pi_1(g)
    \punkt
\end{equation*}
If $G\cong \ZZ^n$ this is true if and only if $\hat c_1=\hat c_2$
(see for example \cite[Lemma 5.5]{Ka1998} cf. also \cite[Lemma A.1.2]{DoHaRo1969II})
where $\hat c(g,h)=c(g,h)c(h,g)^{-1}$ is the \emph{commutator map} or 
\emph{antisymmetric part} of a cocycle $c$. The following Lemma 
will be useful showing the equivalence of two cocycles. 
\begin{lem}\label{lem:cocycle} Let $G=G_1\times G_2$ be an Abelian group 
    and $c,c'\in Z^2(G,\TT)$ be two 2-cocycle and 
    $c_i, c_i'\in Z^2(G_i,\TT)$ their restrictions to 
    $G_i\times G_i$ for $i=1,2$. 
    If $[c_i]=[c_i'] \in H^2(G_i,\TT)$ then 
    $\hat c =\hat c'$ implies $[c']=[c] \in H^2(G,\TT)$.
\end{lem}
\begin{proof} The proof is basically 
    \cite[Proof of Lemma A.1.2.]{DoHaRo1969II}).
    Because $c\longmapsto \hat c$ is a homomorphism it is enough 
    to show that for $c\in Z^2(G,\TT)$: if 
    $\hat c=1$ and $c_i(g_i,h_i)=b_i(g_i)b_i(h_i)/b_i(g_ih_i)$
    then $c\in B^2(G,\TT)$, i.e 
    $c=\delta b$ with 
    $b\in Z^1(G,\TT)$.
    Indeed, setting     $$
        b(g_1g_2)=\frac{b_1(g_1)b_2(g_2)}{c(g_1,g_2)}
        \equiv\frac{b_1(g_1)b_2(g_2)}{ c(g_2,g_1)}
    $$ for $g_i\in G_i$ we calculate using the cocycle relation:
    \begin{align*}
        c(g_1g_2,h_1h_2) &=\frac{c(g_1,g_2h_1h_2)c(g_2,h_1h_2)}{c(g_1,g_2)}
        \\
        &= \frac{c_1(g_1,h_1)c(g_1h_1,g_2h_2)c(g_2h_2,h_1)c_2(g_2,h_2)}
        {c(g_1,g_2)c(h_1,g_2h_2) c(h_2,h_1)}
        \\
        &= \frac{b(g_1g_2)b(h_1h_2)}{b(g_1h_1g_2h_2)}
        \punkt
    \end{align*}
\end{proof}

Equivalently to say that $\pi$ is a projective representation there is a  true 
representation also denoted by $\pi$ of the group $\tilde G = G\times \TT$ with
multiplicative law $(g_1,t_1)(g_2,t_2) =(g_1g_2, c(g_1,g_2) t_1t_2)$ given by 
$\pi(g,t) = t \pi(g)$.  One calls $\tilde G$ a \emph{central extension of} $G$.

Let $G$ be a Lie group and $\pi$ a continuous projective unitary 
representation of $\Lp G$ on a Hilbert space $\Hil$.  We assume that there is an
action of the rotation, i.e. $\TT$ acts unitarily on $\Hil$ by $U$ such that 
$U(\theta)\pi(f)U(\theta)^\ast = \pi(R_\theta f)$ where 
$R_\theta f(\theta') = {f(\theta'-\theta)}$ for $f\in \Lp G$. 
In other words we assume $\pi$ extends to a representation of 
$\Lp G\rtimes \TT$.
Then $\pi$ is called  \emph{positive energy} (cf. \cite{Se1981,PS1986}\footnote{we use a different convention, which fits with the definition of positive energy for conformal nets}) if 
\begin{align*}\Hil &= \oplus_{n\geq 0} \Hil_n, &
     \Hil_n 
=\{x\in\Hil : U(\theta)x=\e^{\ima n\theta}\} 
\end{align*} 
with $\dim \Hil_n <\infty$ and\footnote{%
    This can 
    be obtained by multiplying a given representation of $\TT$ 
    with a character of $\TT$
} $\Hil_0\neq\{0\}$. 
That means the generator $L_0$ of $U(\theta)=\e^{\ima \theta L_0}$
has positive spectrum.

Let $\Lpc F$ be the central extension of $\Lp F$
defined by the cocycle
\begin{equation*}
    c_F(f,g)  = \e^{-\ima \omega(f,g)}
    = \e^{-\ima/2 \int \langle f, g'\rangle}.
\end{equation*}
Then the conformal net $\A_F$ constructed above can be regarded as 
the conformal net associated with a positive energy representation of $\Lp F$
with cocycle $c_F$ or equivalently a (true) positive energy representation 
of $\Lpc F$.
For $I\in\cI$ we denote by $\Lp_I F$
all loops with support in $I$.
\begin{prop}
    Let $\Hil_{0,F}=\overline{\Lp F/F}$ be the one-particle space associated 
    with $(F, \langle \slot,\slot\rangle)$.
    There is a unitary positive energy representation $\pi$ of $\Lpc F$ on the Fock space $\Hil_F\equiv\ e^{\Hil_{0,F}}$ given by
    $\pi_0:\Lpc F\equiv \Lp F \times \TT \ni(f,c) \longmapsto c\cdot W([f])$, where $[f]\in \Lp F/F\subset \Hil_F$. 
    In particular it is $\A_F(I)=\pi_0(\Lp_I F)''$.
\end{prop}
\begin{proof} $W([f])$ is unitary by construction. Obviously by the Weyl commutation relations $\pi$ is a representation 
    of $\Lpc F$ with the given cocycle. Let $L_0$ be the positive generator
    of the rotations on $\Hil_{0,F}$. The generator of the rotation on $\Hil_F$ 
    is then given 
    by $\tilde L_0=1\oplus L_0 \oplus (L_0\otimes 1 + 1 \otimes L_0) \oplus \cdots$ and in 
    particular positive.
    
    Further for $I\in\cI$ by construction $\pi(\Lp_I F)''=R(\Lp_I F/F)$ which equals $\A(I)\equiv R(H_F(I))$ by 
    Proposition \ref{prop:CCR} and because $\Lp_I F$ is dense in $H_F(I)$
    again by construction.
\end{proof}

\subsection{Conformal Nets Associated with Lattices}\label{sec:ChLaNet}
We want to consider local extensions of the net $\A_F$ associated with 
Abelian currents with values in $F$.
The case of $F=\RR$ (one current) was treated in \cite{BuMaTo1988} and 
the extensions are given by a charge $g= \sqrt{2 N}$ with $N \in \NN$. 
The general case was elaborated in \cite{St1995} with the result that
the extensions are given by even integral lattices (for $n=1$ the lattice is $g\ZZ$). 
The same lattice models were also examined in \cite{DoXu2006},
where they are equivalently defined as a positive energy representation of the 
loop group of the torus associated with the lattice. This gives a connection to
the representations of loop groups at level 1 \cite{Se1981,PS1986}
for simply laced Lie groups.
The lattice models are well known in the framework of vertex operator 
algebras.
For a treatment of lattice models in vertex operator algebras and its 
connection to Kac--Moody algebras 
we refer \eg to \cite[Chapter 5.4]{Ka1998}. 

Let $L$ be an \emph{integral (positive) lattice}, i.e. a free $\ZZ$-module 
with positive-definite integral bilinear form 
$\langle \slot,\slot\rangle:L\times L \longrightarrow \ZZ$. 
A lattice is called \emph{even} if $\langle \alpha,\alpha\rangle \in 2\NN$ 
for all $\alpha \in L$, and we note that an even lattice is necessarily 
integral. 
To a lattice $L$ we relate an Euclidean space 
$(F,\langle \slot,\slot\rangle)$ where $F=L\otimes_\ZZ \RR$ and the scalar 
product $\langle \slot,\slot\rangle$ is continued 
to $F\times F \longrightarrow \RR$ by 
linearity. The dimension $n=\dim F$ is called the \emph{rank} 
(assumed to be finite).

Equivalently, we can view an even lattice $L$ as
a free discrete subgroup of a finite dimensional Euclidean space 
$(F,\langle \slot,\slot\rangle)$ which spans $F$ 
and satisfies 
$\langle \alpha,\alpha \rangle \in 2\NN$ for all $\alpha \in L$.
Let $L$ be even and $L^\ast := \lbrace x\in V : \langle x, L \rangle \subset \ZZ\rbrace$ be the dual lattice \cite{CoSl1998}. 
It is a not necessarily an integer lattice and 
can canonically be identified  
with $\Hom(L,\ZZ)$ by the scalar product. 
It is $L\subset L^\ast$ and it can be shown that the group 
$L^\ast/L$ is finite. In the case $L^\ast =L$ the lattice is called
    \emph{self-dual} or 
    \emph{unimodular}
and in this can be the case only for rank $n\in 8\NN$.

With an even lattice $L$ we associate a torus $T=F/2\pi L$, and we will 
represent elements by $\e^{\ima f}$ with $f\in F$ and $\e^{\ima f}=1$ if
and only if $f\in L$, formally
\begin{align*}
    F / 2\pi L  &\longbijects T,
    ~    [t] \longmapsto \e^{\ima t}
    \punkt
\end{align*}

\subsubsection{Loop group associated with a torus}
Let $\Lp T=C^\infty(\Sc, T)$ the loop group associated with the torus $T$.
We write  $\e^{\ima f}$ for an element in $\Lp T$ where we mean the function 
$\e^{\ima\theta} \longmapsto 
\e^{\ima f(\theta)}$ and $f:\RR \longrightarrow F$ is a smooth function such 
that
the winding number 
\begin{equation*}
    \Delta_f:= \frac{1}{2\pi} (f(\theta + 2\pi) -f(\theta))
\end{equation*} 
is constant and takes values in $L$. 
In particular $f_\circ:{\theta} \longmapsto f(\theta) - \Delta_f\cdot \theta$ 
is a periodic function and we can decompose 
\begin{equation*}
    f(\theta) = \Delta_f\cdot \theta + f_0 + \sum_{n\in \ZZ^\ast} f_n \e^{\ima n\theta}
\end{equation*}
where we call $f_0$ the \emph{zeroth-mode}.

We are interested in projective positive energy representations of $\Lp T$ or equivalently representations of a central extension:
\begin{align*}
    1 \longrightarrow \TT \longrightarrow \Lpc T\longrightarrow \Lp T \longrightarrow 1 
\end{align*}
which are given by a cocycle $c:\Lp T \times \Lp T \longrightarrow \TT$
specified in the following.

It is well-known (see for example \cite{Ka1998}) that there existss a bilinear 
form $b:L\times L \longrightarrow \ZZ_2$ such that
$$
    b(\alpha,\alpha) = \frac 12 \langle \alpha,\alpha\rangle 
    \quad \text{for all } \alpha \in L
    \komma
$$
\eg if $\{\alpha_1,\ldots,\alpha_n\}$ is a basis of $L$ one can choose
\begin{align*}
    b(\alpha_i,\alpha_j) &= 
        \begin{cases} 
            \langle \alpha_i,\alpha_j \rangle \mod 2 & i<j\\
            \frac 12 \langle \alpha_i,\alpha_i\rangle \mod 2 & i=j\\
            0 &i>j 
            \punkt
        \end{cases}
\end{align*}
Therefore a bimultiplicative map
$\varepsilon(\alpha,\beta) :L\times L \longrightarrow \lbrace+1,-1\rbrace \cong \ZZ_2$ exists, 
satisfying
$\varepsilon(\alpha,\alpha) = (-1)^{\langle \alpha,\alpha\rangle/2 }$. 
Such a map is a 2-cocycle satisfying:
\begin{align*}
    \varepsilon(\alpha,\beta+\gamma)\varepsilon(\beta,\gamma) 
        &=\varepsilon(\alpha,\beta)\varepsilon(\alpha+\beta,\gamma)\\
    \varepsilon(\alpha,\beta)\varepsilon(\beta,\alpha)
        &=(-1)^{\langle\alpha,\beta\rangle}
        \punkt
\end{align*}
Now we specify the central extension $\Lpc T$ by choosing  a 2-cocycle $c:\Lp T \times \Lp T \longrightarrow \TT$ as in \cite{Se1981}
\begin{align}
    \label{eq:cocy}
    c(\e^{\ima f}, \e^{\ima g}) \equiv c(f,g) &= \varepsilon(\Delta_f,\Delta_g)\e^{\ima S(f,g)}\\
    2\cdot S(f,g) &=  \int_0^{2\pi} \langle f'(\theta), g(\theta)\rangle \frac{\dd\theta}{2\pi}
    + \langle \Delta_f, g(0)\rangle
    \punkt
    \nonumber
\end{align}
We note that the central extension (up to equivalence) 
does not depend on the explicit choice of the 2-cocycle in its equivalence 
class. Further 
we write the relations in $\Lpc T$ formally as $\e^{\ima f} \e^{\ima g} = c(\e^{\ima f}, \e^{\ima g})\e^{\ima (f + g)}$.
It is straightforward  to verify the following relations.
\begin{lem}[\cf \cite{DoXu2006}] \label{lem:LGrel}
    Let $\e^{\ima f},\e^{\ima g} \in \Lp T$, then
    we have the following relations in $\Lpc T$: 
\begin{align*}
    \e^{\ima f}\e^{\ima g} \left(\e^{\ima f}\right)^{-1} &= 
    \e^{\ima \pi \langle\Delta_f,\Delta_g\rangle}
    \e^{\ima \int \langle f'_1,g_1\rangle}\e^{\ima \langle\Delta_f, g_0\rangle -\ima\langle \Delta_g, f_0\rangle}
    \e^{\ima g}
\end{align*}
\end{lem}

\begin{proof}
    We observe that $(\e^{\ima f})^{-1}=
    c(\e^{\ima f},\e^{-\ima f})^{-1} \e^{-\ima f} =  c(f,f)\e^{-\ima f}$
    and we get:
    \begin{align}
        \nonumber \e^{\ima f} \e^{\ima g}(\e^{\ima f})^{-1} &= c(f,g)c(g,-f)
        c(f,-f)c(f,f) \e^{\ima g}
        \\\nonumber &= c(f,g)c(g,f)^{-1} \e^{\ima g}
        \\&= (-1)^{\langle \Delta_f,\Delta_g\rangle} \e^{\ima(S(f,g)-S(g,f))} 
        \e^{\ima g}
        \punkt
        \label{eq:cocyc1}     
    \end{align}
    Using $f(\theta)=\Delta_f\cdot \theta + f_0 + f_1(\theta)$ and $f(\theta)=\Delta_g\cdot \theta +g_0 + g_1(\theta)$
    we have
    \begin{align*}
        \label{eq:Sexplicit}
        S(f,g) &= \frac{1}{4\pi} \int_0^{2\pi} \langle f_1'(\theta),g_1(\theta)\rangle \dd \theta
        + \frac \pi 2 \langle\Delta_f,\Delta_g\rangle + \frac12 \langle \Delta_f, g_0\rangle+ {}\\ &
        +\frac12 \langle f_1(2\pi),\Delta_g\rangle + \frac12 \langle \Delta_f,g_1(0)\rangle
    \end{align*}
    and this gives
    \begin{align*}
        S(f,g)-S(g,f) 
        &= \frac{1}{2\pi} \int_0^{2\pi} \langle f_1'(\theta),g(\theta)\rangle \dd \theta
        + \frac 12 \langle \Delta_f, g_0\rangle
        - \frac 12 \langle \Delta_g, f_0\rangle
        \\&+\frac12 \langle \Delta_f, g(0)-g_1(0)\rangle
        -\frac12 \langle \Delta_g, f(0)-f_1(0)\rangle
        \\&= \frac{1}{2\pi} \int_0^{2\pi} \langle f_1'(\theta),g_1(\theta)\rangle \dd \theta
        + \langle \Delta_f, g_0\rangle
        - \langle \Delta_g, f_0\rangle
    \end{align*}
    which inserted in \eqref{eq:cocyc1} completes the proof. 
\end{proof}

The following Proposition is proved in \cite[Proposition 3.4]{DoXu2006} and
shows that the central extension is local, i.e. loops supported in disjoint
intervals commute.
\begin{prop}[{Locality cf. \cite[Prop. 3.4]{DoXu2006}}] If $\supp \e^{\ima f}\cap \supp \e^{\ima g} = \emptyset$
    then $\e^{\ima f}\e^{\ima g} = \e^{\ima g } \e^{\ima f}$.
    \label{prop:locality}
\end{prop}
For $I\in\cI$ we denote by 
$\Lp_I T = \{ \e^{\ima f} \in \Lp T : \supp \e^{\ima f} \subset I\}$ 
all loops with support in $I$ and by $\Lpc_I T$ the preimage of $\Lp_I T$ 
under the covering map. 
Using this locality and well-known results of positive energy representations
of loop groups it is shown in \cite{DoXu2006} that there is a conformal 
net associated with $\Lpc T$, precisely: 
\begin{prop}[Local conformal net associated with $\Lpc T$ cf. \cite{DoXu2006}] 
    There is a correspondence between the elements of $L^\ast/L$ and 
    positive energy representation of $\Lpc T$. Let $\pi_{(L,0)}$ be the (vacuum) representation 
    corresponding to $[0]\in L^\ast /L$, then
    \begin{align*}
        I \longmapsto \A_{\Lp T}(I):= \pi_{(L,0)}(\Lpc_I T)'' 
    \end{align*}
    is a completely rational conformal net with $\mu$-index $\mu=|L^\ast/L|$ 
    and has $\mu$ sectors of statistical dimension $1$ corresponding to the 
    positive energy representations of $\Lpc T$.
\end{prop}
\begin{proof}
    For the first statement see \cite[Lemma 3.5]{DoXu2006}
    and \cite[Section 9.5]{PS1986}.
    $\A_{\Lp T}$ is a local net by Proposition \ref{prop:locality}
    and \cite[Proposition 3.1]{DoXu2006}
    shows that it is a strongly additive conformal net fulfilling
    the split property. 
    \cite[Proposition 3.15]{DoXu2006} shows the correspondence between sectors
    and elements of
    $L/L^\ast$  and the $\mu$ index is given in 
    \cite[Corollary 3.19]{DoXu2006}.
\end{proof}
\begin{rmk}
We note that the construction depends only on $L$ and we denote this net also 
by $\A_L$, the \emph{conformal net associated with the lattice $L$}.
\end{rmk}

In the rest of the section we give the construction in a more explicit manner. 
In particular the Hilbert space $\Hil_L$ of $\A_L$ can naturally be
identified with $L$ copies of the Hilbert space $\Hil_F$ of $\A_F$ 
(more precisely $\ell^2(L,\Hil_F)$) where $F=L\otimes_\ZZ \RR$.
This enables us to show that $\A_L(I)$ is a crossed product of 
$\A_F(I)$ with $L$.

The identity component $(\Lpc T)_0$ of $\Lpc T$ can be identified with 
$H_F \times T$, 
where $H_F=\Lp F/F\times \TT$ is the Heisenberg group with 
multiplication law $(f,c_1)(g,c_2)=(f+g,\e^{-\ima/2\omega(f,g)}c_1c_2)$.
The representation $\pi_0$ of $\Lpc F$ is a representation of $H_F$ because 
the constant loops lie in the kernel of the representation and it turns out
to be the unique irreducible representation 
(cf. proof of 9.5.10 \cite{PS1986}) with positive energy. 
Let $\tilde W = H_F\times F$ (the idea is to add an operator $Q$ which 
measures the charge). 
All irreducible representations of positive energy of $\tilde W$ are 
classified by a charge $\alpha \in F$ and are of the form
$(\pi_\alpha, (\Hil_F)_\alpha)$ given by 
$\pi_\alpha(f,v) = \e^{-\ima 2\pi \langle \alpha, v\rangle} \pi_0(f)$.
As a set it is $\Lpc T \equiv (\Lpc T)_0 \times L$ and it is shown in 
\cite{PS1986} that all irreducible representations of $\Lpc T$ of positive 
energy are given by points $\lambda\in L^\ast/L$ and are acting on the 
Hilbert space
\begin{equation*}
    \Hil_{(L,\lambda)} =\bigoplus_{\alpha\in \lambda + L} (\Hil_F)_{\alpha}
    \punkt
\end{equation*}

The Hilbert space $\Hil_{(L,0)}$ on which $\Lpc T$ acts is graded by the 
lattice $L$, and we call $\alpha$ the charge of the subspace 
$(\Hil_F)_{\alpha}$ of $\Hil_{(L,0)}$. 
We define for $\alpha\in L$ charge shift operators $\Gamma_\alpha$ by 
$(\Gamma_\alpha x)_\beta = x_{\beta-\alpha}$ and introduce the unbounded 
charge operator $Q$ satisfying $(Qx)_\alpha = \alpha x$.
The $\Gamma_\alpha$ does not fulfill exactly the commutation relations 
suitable for the representation of $\Lpc T$. But the commutation relations 
between different $\Gamma_\alpha$ can be changed by a so called Klein 
transformation. 
Let $\eta:L\times L \longrightarrow \TT$ be a bimultiplicatice map (2-cocycle) 
and  $\tilde\Gamma_{\alpha} = \eta(-Q,\alpha)\Gamma_{\alpha}$ then 
we have:
\begin{lem}\label{lem:projL}
    $\alpha \longmapsto \tilde\Gamma_{\alpha}$ defines a representation 
    of the central extension $\tilde L$ of $L$ by the cocycle 
    $\eta(\slot,\slot)$.
\end{lem}
\begin{proof} 
    We note that $\tilde L = L \times \TT$ with multiplication law 
    $$(\alpha,c)(\beta,d) = (\alpha+\beta,\eta(\alpha,\beta)cd)$$ and the
    representation is obtained by applying 
    $(\alpha,c) \longmapsto c \Gamma_\alpha$. 
    Indeed we calculate
    \begin{align*}
        \tilde\Gamma_\alpha\tilde\Gamma_\beta 
            &= \eta(-Q,\alpha)\Gamma_\alpha \eta(-Q,\beta)\Gamma_\beta
        \\  &= \eta(-Q,\alpha)\eta(-Q+\alpha,\beta)\Gamma_\alpha \Gamma_\beta
        \\  &= \eta(\alpha,\beta) \eta(-Q,\alpha+\beta)\Gamma_{\alpha+\beta}
        \\  &= \eta(\alpha,\beta)\tilde \Gamma_{\alpha+\beta}
        \punkt \qedhere
    \end{align*}
\end{proof}
We choose $\eta(\alpha,\beta)= c(\e^{\ima t_\alpha},\e^{\ima t_\beta})$ 
where $t_\alpha(\theta) =\alpha\cdot \theta$ and get a representation 
of $\{(\e^{\ima t_\alpha},c) : \alpha \in L\} \subset \Lpc T$
by $(\e^{\ima t_\alpha},c) \longmapsto c\cdot \tilde\Gamma_{\alpha}$.
\begin{prop} 
    The vacuum representation of $\Lpc T$ acts by the above construction 
    irreducible on 
    \begin{equation*}
        \Hil_{L}:=\Hil_{(L,0)}\equiv \bigoplus\limits_{\alpha\in L}
        (\Hil_F)_\alpha
    \end{equation*}
    i.e. the local net $\A_L$ acts on $\Hil_{L}$. 
\end{prop}
\begin{proof}
    Let $\e^{\ima f}, \e^{\ima g} \in \Lp T$.
    We note first that for $f=f_\Delta +f_0+f_1$ with 
    $f_\Delta(\theta)=\Delta_f\cdot\theta$ and $f_0$ zero-mode like before, 
    we have $\e^{\ima f} = k\e^{\ima f_0}\e^{\ima f_1} \e^{\ima f_\Delta}$ 
    with an irrelevant phase 
    $k=\e^{\ima/2 \langle f_1(2\pi),\Delta_f\rangle}\in\TT$. 

    We claim that 
    \begin{align*}
        \pi'(\e^{\ima f}) = \e^{-\ima \langle f_0,Q\rangle} W(f_1) 
        \tilde\Gamma_{\Delta_f}
    \end{align*}
    defines a projective representation of $\Lp T$ with a to $c(\slot,\slot)$
    equivalent cocycle $c'(\slot,\slot)$.
    Then there exists a coboundary $b_h(\slot,\slot)$ like in \eqref{eq:cob} 
    with $c(f,g)=b_h(f,g)c'(f,g)$ and $\pi(f) = h(f)\pi'(f)$ is the wanted 
    representation.

    We can write $\Lp T = \Lp T_0 \times L$  where $\Lp T_0$ is the connected
    component of the identity and $\alpha \in L$ is identified with the loop 
    $t_\alpha(\theta)=\alpha\cdot \theta$. 
    The cocycles restricted to $L$ are equivalent by Lemma \ref{lem:projL}.
    Further the Weyl relations give exactly the relations of the cocycle 
    $c(\slot,\slot)$, namely
    \begin{align*}
        \pi(\e^{\ima (f_0+f_1)}) \pi(\e^{\ima (g_0+g_1)}) 
        &= \e^{-\ima \langle f_0,Q\rangle} W(f_1) \e^{-\ima \langle g_0,Q\rangle} 
            W(g_1)
        \\&= \e^{-\ima \langle f_0 +g_0,Q\rangle} 
            \e^{\ima/2 \int \langle f_1',g_1\rangle} 
            W(f_1 + g_1)
        \\&=\e^{\ima/2 \int \langle f_1',g_1\rangle} 
            \pi(\e^{\ima (f_0+g_0+f_1+g_1)})
        \\&=c(f_0+f_1,g_0+f_1) \pi(\e^{\ima (f_0+g_0+f_1+g_1)})
        \komma
    \end{align*}
    so the cocycles restricted to $\Lp T_0$ are also equal. 
    By Lemma \ref{lem:cocycle} it is sufficient to check that the pairwise 
    commutation relations of  $\pi'(\e^{\ima f})$ and $\pi'(\e^{\ima g})$ equal 
    the one of $\Lpc T$ given in Lemma \ref{lem:LGrel}, indeed
    \begin{align*}
        \pi(\e^{\ima f_\Delta}) \pi(\e^{\ima (g_0+g_1)}) 
            \pi(\e^{\ima f_\Delta})^\ast
        &= \tilde \Gamma_{\Delta_f} \e^{-\ima \langle g_0,Q\rangle} W(g_1) 
            \tilde\Gamma_{\Delta_f}^\ast 
        \\&=  \e^{-\ima \langle g_0,Q-\Delta_f \rangle} W(g_1)
        \\&=  \e^{\ima \langle \Delta_f,g_0\rangle} \pi(\e^{\ima (g_0+g_1)})
        \punkt
        \qedhere
    \end{align*}  
\end{proof}
\begin{prop} 
    The local algebras $\A_L(I)$ are given by a crossed product of $\A_F(I)$ 
    with $L$.
\end{prop}
\begin{proof} Let $I$ be a proper interval and 
    $y \in \Sc \setminus \overline I$. 
    The local loop group $\Lpc_I T$ is generated
    by loops $\e^{\ima f}$ with $f(x) \in 2\pi L$
    for $x \not\in I$. We note that $\Lpc_I T= (\Lpc_I T)_0 \times L$ as a set, 
    where $(\Lpc_I T)_0$ is the connected component of the identity consisting 
    of loops 
    $\e^{\ima f}$ with $\Delta_f=0$ and $L$ is identified with 
    $\{\e^{\ima f_\alpha} : \alpha\in L\}$,
    where $t_\alpha$ are functions like above with $\Delta_{t_\alpha}=\alpha$.
    We choose a basis $\{\alpha_i\}$  of $L$ and some smooth 
    ``step function''
    $M:\RR \longrightarrow \RR$ with $M(\theta +2\pi)=M(\theta)$ and with 
    $\Delta_M=1$, such that for $x\not\in I$ it is
    $M(x) \in \ZZ$ and therefore $m(x):=M'(x)=0$. 
    The loop $\e^{\ima M\alpha_i}$ has winding number $\alpha_i$
    and implements an automorphism $\beta_i$ of $\pi( (\Lpc T)_0 )''$ 
    \begin{align*}
        \pi(\e^{\ima f}) &= \e^{-\ima \langle Q, f_0\rangle} 
            W([f_1])=:\tilde W (f) \\
        \beta_i &:= \Ad \pi(\e^{\ima M\alpha_i})
        \\
        \beta_i (\tilde W(f)) &= \e^{\ima \int \langle f,m \alpha_i \rangle}
            \tilde W(f).
    \end{align*}
    which defines an automorphic action $\beta$ of $L$ on the algebra 
    $\pi( (\Lpc T)_0 )''$.
    We note that with the notation from above 
    $\Hil_F \cong (\Hil_F)_0=\overline{\pi( (\Lpc T)_0) \Omega}
    \subset \Hil_{(L,0)}$ 
    and denote by
    $\pi_F:(\Lpc T)_0 \longrightarrow \U( (\Hil_F)_0)$ 
    the representation of $(\Lpc T)_0$ on $(\Hil_F)_0$ obtained 
    by restriction of $\pi$.
    By construction we get $\A_F(I) = \pi_F( (\Lpc_I T)_0 )''$.
    This is the vacuum representation and it is $W([f])=\tilde W(f)$. 
    Finally we can see $\beta_i$ as an automorphism of $\A_F(I)$
    \begin{align*}
        \beta_i (W([f])) &= \e^{\ima \int \langle f-f(y),
        m \alpha_i \rangle} W([f])   
    \end{align*}
    and it is clear that 
    $\pi(\Lpc_I T)'' = \A_F(I) \rtimes_\beta L$, where the action is free and 
    faithful.
\end{proof}
\begin{rmk} By construction we have that 
    $\A_{L\oplus Q} \cong \A_L \otimes \A_Q$.
    \label{rmk:irrlattices}
\end{rmk}
The adjoint action of a (localized) loop $\e^{\ima f}$ with  
$\Delta_f=\lambda \in L^*$ gives a localized endomorphism
of $\A_L$ which belongs to the sector $[\lambda]\in L^\ast/L$.
The conformal spin is well known to be 
$\e^{\ima\pi\langle \lambda,\lambda\rangle}$. 
\begin{prop} Let $L\subset Q$ be two even lattices of the same rank $n$. 
    Then the local conformal net
    $\A_Q$ is the simple current extension (see \cite[Lemma 2.1]{KaLo2006}) 
    of $\A_L$ by  
    the subgroup $Q/L$ of the group $L^\ast/L$ of all sectors of 
    $\A_L$. 
    \label{prop:sublattices}
\end{prop}
\begin{proof}
    By construction it is $\A_L\subset \A_Q$ and $\Hil_L\subset \Hil_Q$. 
    Let us denote by $\cB$ the net obtained by 
    the simple current extension 
    by $Q/L$, which is the crossed product with automorphisms given by the
    adjoint action by 
    loops $\e^{\ima f}$  with $\Delta_f\in Q/L$. So clearly we can see
    $\cB$ as conformal subnet
    of $\A_Q$
    and they coincide because $\overline{\cB(I)\Omega}=\Hil_Q$.
\end{proof}
\begin{rmk} In \cite{KaLo2006} another construction of lattice models is given,
    which starts with a conformal net $\A$, which is the simple 
    current extension by the 
    dimension 1 sector of $\Vir_{c=1/2} \otimes \Vir_{c=1/2}$. 
    In \cite[Remark 2.3]{KaLo2006} the authors
    conjecture that $\A$ is a Buchholz--Mack--Todorov extension, namely the
    one with $g=2$ which is in our language the conformal net 
    $\A_{2\ZZ}$, where $2\ZZ$ is the lattice
    with $\langle\alpha,\beta\rangle =\alpha\beta$. 
    Let us assume this conjecture is true. 
    They take even lattices $L\supset (2\ZZ)^n$ (coming from codes) 
    and take the simple 
    current extension by the group  
    $L/(2\ZZ)^n \subset (1/2\ZZ)^n /(2\ZZ)^n$
    of the net $\A^{\otimes n}$, which  
    is under the conjecture 
    isomorphic to $\A_{(2\ZZ)^n}$ using Remark
    \ref{rmk:irrlattices}. By Proposition \ref{prop:sublattices} 
    the extended net  
    then would be 
    isomorphic to our net $\A_L$ and the two constructions would coinside.
\end{rmk}

\subsection{Loop Group Models of Simply Laced Groups at Level 1}\label{sec:ChLGNet}
We show the relation of the lattice model associated with the root lattice 
$L$ of simply laced group $G$ to the level 1 representation of the loop group 
of  $\Lp G$ \cite{PS1986,Se1981,St1995}.

Let $G$ be a 
    compact, connected, simply connected, simply laced Lie group 
with maximal torus $T$.
\emph{Simply laced} means that there is an invariant inner product on its Lie 
algebra $\mathfrak g$ for which all roots have the same length or 
equivalently the Weyl group of $G$ acts transitively on the roots.
By \cite[Theorem 3.12]{GaFr1993}
the vacuum positive energy representation $\pi$ of $\Lp G$
at level $k$ gives rise to a conformal net denoted by $\A_{G_k}$, defined by
$\A_{G_k}(I) = \pi(\Lp_I G)''$. 

Let $\mathfrak{t}$ be the Lie algebra of $T$ and let us  identify 
\begin{align*} 
    \mathfrak t / 2\pi L  &\longbijects T\subset G,  
    [t] \longmapsto \e^{t}
    \punkt
\end{align*}
The roots of $G$ are linear maps $\alpha:\mathfrak t \longrightarrow \RR$. For
each $\alpha$ we define a $h_\alpha\in\mathfrak t$ such that 
$\alpha(t)=\langle h_\alpha,t\rangle$ for $t\in\mathfrak t$ where 
$\langle\slot,\slot\rangle$ is the Cartan--Killing form which we can assume to 
be normalized such that $\langle h_\alpha,h_\alpha\rangle = 2$.
This can 
always be realized, due to $G$ being simply laced. 
In the case $\SU(N)$ and $\Spin(2N)$ it is given explicitly by 
$\langle x,y\rangle =-\tr(xy)$ and $\langle x,y\rangle =-1/2\tr (xy)$, 
respectively.
It is well known that $h_\alpha \in L$ and that the set of the $h_\alpha$ with
$\alpha$ a root coincide with the $x\in L$ such that $\langle x,x\rangle =2$.
By abuse of notation we identify $\alpha$ with $h_\alpha \in L$, \ie
$\langle \alpha,\beta\rangle \equiv \langle h_\alpha,h_\beta\rangle$.
We note the missing $\ima$ in the $\exp$ map due to the conventions 
$t^\ast = -t$ for $t\in\mathfrak t$, \ie by identifying $F=-\ima \mathfrak t$
we get the relation to the former notation. 

Let $\pi$ be a positive energy representation of $\Lp T$
with cocycle \eqref{eq:cocy}, where $T$ is the torus associated with $L$ but 
also a maximal torus of $G$ by the above discussion. 
We note that the cocycle of the level 1 representation of $\Lp G$ 
restricted to $\Lp T$ is (equivalent to) our cocycle \eqref{eq:cocy} by
\cite[Proposition 4.8.3]{PS1986}.
More remarkable is the following result by Segal \cite{Se1981}, 
stating that the representation of $\Lp T$ extends to $\Lp G$. 
This is mainly achieved by taking a limit of loops with winding number
$\Delta_f=\alpha$, and building so called ``vertex'' or ``blib'' operators
which turn out to generate---together with the generators of loops 
with trivial winding number---a representation
of the polynomial algebra $L^\mathrm{alg} \mathfrak g$,
which is then exponentiated.
\begin{prop}[{\cite[Proposition 4.4]{Se1981}}]\label{prop:Segal}
    Let $G$ be a 
        compact, connected, simply connected, simply laced Lie group 
    with maximal torus $T$. 
    If $\pi$ is a positive energy projective representation of $\Lp T$ with 
    cocycle above, then the action of $\Lp T$ extends canonically to an action 
    of $\Lp G$.
\end{prop}
Now we want to apply this result to show that certain loop group nets at level 
1 are a special case of the conformal nets associated with lattices. The analog 
of the following result is well known in the theory of vertex operator algebras
under the name Frenkel--Kac or Frenkel--Kac--Segal construction.
\begin{prop}[Algebraic version of the Frenkel--Kac--Segal construction]
    Let $G$ be a compact, simple, connected, simply connected and simply laced 
    Lie group and $L$ its root lattice as above. Then the conformal net $\A_L$ 
    is equivalent to the loop group net $\A_{G,1}$ at level 1 associated with 
    $\Lp G$.
    In particular $\A_{G,1}$ is completely rational and has 
    $\mu$-index $\mu=|L^\ast/L|$. 
\end{prop}
The case $G=\SU(N)$ is stated in \cite[3.1.1]{Xu2009}
and the general case in \cite{St1995}*{p. 37/38}.
In principle we could 
try to use the result of Segal to directly proof the proposition, but
locality of the constructed exponentiated currents is not clear and
has to be checked. Therefore we give a more indirect proof using an operator 
algebraic argument.
\begin{proof} 
    Let $\pi$ be the vacuum positive energy representation at level 1 of $\Lp G$
    and by Proposition \ref{prop:Segal} it can be assumed to act on the Hilbert
    space $\Hil_L$.
    We see $\pi$ as a representation of the central extension $\Lpc G$. 
    It is $\Lpc T
    \subset \Lpc G$ and in particular for every $I\in\cI$ also $\Lpc_I T \subset
    \Lpc_I G$. 
    This implies that $\A_L(I)\equiv\pi(\Lpc_I T)''$ is a conformal subnet of 
    $\A_{G,1}(I)=\pi(\Lpc_I G)''$.
    Because $\overline{\A_L(I)\Omega}=\Hil_L$
    by Lemma \ref{lem:subnet} the two nets $\A_L$ and $\A_{G,1}$ have to coincide.
\end{proof}
\begin{example}
    The simple, simply laced groups correspond to 
    the Dynkin diagrams of type A, D and E (see Figure \ref{fig:dynkin}), 
    namely 
    $\SU(n+1)$ for $A_n$ with $n\geq 1$,
    $\Spin(2n)$ for $D_n$ with $n\geq 4$ 
    and in the exceptional case the compact, simply connected 
    forms of $E_6,E_7,E_8$. The level 1 loop group nets of these 
    groups are therefore given by lattice models of their root lattice $L$,
    which is characterized by the basis 
    $\{\alpha_1,\ldots,\alpha_n\}$  with $n$ the rank of $L$ and
    $\alpha_i$ represents the $i$-th vertex of the Dynkin diagram. The
    inner product is specified by the \emph{Cartan matrix} $(C_{ij})$
    via
    $$
    \langle \alpha_i,\alpha_j\rangle = C_{ij}=
        \begin{cases} 
            2 & i=j \\ 
           -1 & $i$\text{ and }$j$\text{ are connected by an edge}\komma\\
           0 & \text{otherwise}\punkt
        \end{cases}
    $$

\end{example}
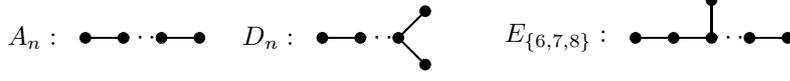
\begin{figure}
            \begin{tikzpicture}[scale=0.5]
                \useasboundingbox (-1,-1.5) rectangle (5,3); 
                \node at  (-0.5,0) [left] {$A_n:$};
                \draw[fill] (0,0) circle (4pt) ;
                \draw[thick] (0,0)--(1,0) ;
                \draw[fill] (1,0) circle (4pt) ;
                \node at (1.5,0) {$\cdots$};
                \draw[fill] (2,0) circle (4pt) ;
                \draw[thick](2,0)--(3,0);
                \draw[fill] (3,0) circle (4pt) ;
            \end{tikzpicture}
            \begin{tikzpicture}[scale=0.5]
                \useasboundingbox (-1,-1.5) rectangle (7,3); 
                \node at  (-0.5,0) [left] {$D_n:$};
                \draw[fill] (0,0) circle (4pt) ;
                \draw[thick] (0,0)--(1,0) ;
                \draw[fill] (1,0) circle (4pt) ;
                \node at (1.5,0) {$\cdots$};
                \draw[fill] (2,0) circle (4pt) ;
                \draw[thick](2,0)--(2.707,0.707);
                \draw[thick](2,0)--(2.707,-0.707);
                \draw[fill] (2.707,0.707) circle (4pt) ;
                \draw[fill] (2.707,-0.707)  circle (4pt) ;
            \end{tikzpicture}
            \begin{tikzpicture}[scale=0.5]
                \useasboundingbox (-1,-1.5) rectangle (5,3); 
                \node at  (-0.5,0) [left] {$E_{\{6,7,8\}}:$};
                \draw[fill] (0,0) circle (4pt) ;
                \draw[fill] (1,0) circle (4pt) ;
                \draw[thick] (0,0)--(2,0)--(2,1)--(2,0) ;
                \draw[fill] (2,0) circle (4pt) ;
                \draw[fill] (2,1) circle (4pt) ;
                \draw[fill] (3,0) circle (4pt) ;
                \draw[fill] (4,0) circle (4pt) ;
                \node at (2.5,0) {$\cdots$};
                \draw[thick] (3,0)--(4,0) ;
            \end{tikzpicture}
    \caption{A, D, E Dynkin diagrams}
    \label{fig:dynkin}
\end{figure}
\section{Boundary Quantum Field Theory -- Nets on Minkowski Half-Plane}\label{sec:BQFT}
In this section we want to construct local nets on Minkowski half-plane
$M_+=\{(t,x) \in \RR^2:x>0\}$ which are time-translation covariant
and which we will call also simply \emph{\BQFTnets}.

Let $I_1, I_2$ be two intervals of the time 
axis such that $I_2 > I_1$ and let us define the \emph{double cone}
\begin{align*}
    \cO = I_1 \times I_2:= \{ (t,x)\in\RR^2 : t-x\in I_1, x+t \in I_2 \}
\end{align*}
like in Figure \ref{fig:doublecone}.
We call such a double cone $\cO=I_1\times I_2$ \emph{proper} if it has a 
positive distance to the boundary, i.e. $\overline {I_1}$ and $\overline {I_2}$ 
have empty intersection; the 
    \emph{set of proper double cones} 
we denote by $\cK_+$.
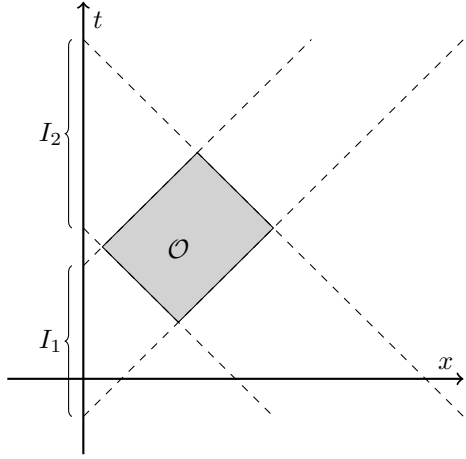
\begin{figure}[ht]
    \definecolor{mygrey}{HTML}{d2d2d2}
    \begin{tikzpicture}
        \draw [thick, ->] (-1,0)--(5,0) node [above left] {$x$};
        \draw [thick, ->] (0,-1)--(0,5) node [below right] {$t$};
        \draw[fill=mygrey] (0.25,1.75)--(1.5,3)--(2.5,2)--(1.25,0.75)
            --(0.25,1.75);
        \draw [dashed] (0,-0.5)--(5,4.5);
    	\draw [dashed] (0,1.5)--(3,4.5);
	    \draw [dashed] (0,2)--(2.5,-0.5);
    	\draw [dashed] (0,4.5)--(5, -.5);
	    \draw[decorate, decoration=brace] (-0.15,-0.5)--node [left] 
            {$I_1$} (-0.15,1.5);
	    \draw[decorate, decoration=brace] (-0.15,2)--node [left] 
            {$I_2$}(-0.15,4.5);
    	\node at (1.25,1.725) {$\cO$};
    \end{tikzpicture}
    \caption{Double cone $\O=I_1\times I_2$ in $M_+$}
    \label{fig:doublecone}
\end{figure}
\subsection{Local Nets of Standard Subspaces on Minkowski Half-Plane}\label{sec:BQFTStdNet}
As an intermediate step we built up local time-translation covariant nets of 
standard subspaces related with the local Möbius covariant nets 
of standard subspaces $H_F$ 
from Proposition \ref{prop:NetStd}
using the semigroup $\cE(H_F(0,\infty))$. 
\begin{defi} 
    By a 
    \emph{local, time-translation covariant net of standard subspaces on $M_+$}
    on a Hilbert space $\Hil$ we mean a family 
    $\{K(\cO)\}_{\cO\in\cK_+}$ of standard subspaces of a Hilbert space $\Hil$  
    which fulfills:
    \begin{enumerate}[{\bf A.}]
         \item\textbf{Isotony.} $\O_1 \subset \O_2$ implies 
            $K(\O_1) \subset K(\O_2)$. 
         \item\textbf{Locality.} If $\O_1,\O_2 \in \cK_+$ are space-like 
            separated then $K(\O_1) \subset K(\O_2)'$.
         \item\textbf{Time-translation covariance.} There is a strongly 
            continuous one-parameter group $U(t)=\e^{\ima t P}$ 
            on $\Hil$ with positive generator $P$, 
            such that 
            \begin{align*}
                U(t)K(\O) &= K(\O_t), &\O \in \cK_+
            \end{align*}
            where $\O_t = \O + (t,0)$ is the in time-direction 
            shifted double cone.
    \end{enumerate}
\end{defi}
\begin{defi}
    Let $F$ be an Euclidean space and $F_\CC=F\otimes_\RR \CC$ its 
    complexification with canonical complex conjugation $x\longmapsto \bar x$. 
    We denote by $\cS_F$ the space of all complex Borel functions 
    $\varphi:\RR \longrightarrow \B(F_\CC)$ which are boundary values of a 
    bounded analytic function $\RR+\ima\RR_+\longrightarrow \B(H)$,
    \ie for $x,y\in F_\CC$ the function $p\longmapsto (x,\varphi(p)y)$ 
    is an analytic Borel function $\RR+\ima\RR_+ \longrightarrow \CC  $ 
    such that $\varphi$ is \emph{symmetric} and \emph{inner},
    \ie  that         $\varphi(-p)=\overline{\varphi(p)}$ and 
    $\varphi(p)\in \U(F_\CC)$ for almost all $p>0$, respectively.
\end{defi}
We note that with $n=\dim F$ the $\cS_F$ space is naturally isomorphic to $\cS^{(n)}$ 
defined in Section \ref{sec:standard}.
\begin{rmk}
    \label{rmk:SemigroupStd}
    We take the standard subspace $H_F:=H_F(0,\infty)$ and $T(t)=U(\tau(t))$ 
    the one-parameter group of translation. 
    Let 
    $\Hil_{0,F}=\Hil_0\otimes_\RR F \cong \bigoplus_{i=1}^n\Hil_0$ from 
    Proposition \ref{prop:NetStd}, which decompose into 
    $n$ copies of the irreducible standard pair $(H_0,T_0)$, i.e. 
    $H_F = H_0\otimes_\RR F \cong \bigoplus_{i=1}^n H_0$.
    Then $\cE(H_F,T)$ can be identified with $\cS_F$ by 
    Theorem \ref{thm:stdred}.
\end{rmk}
\begin{thm}
    Let $H$ be a local Möbius covariant net of standard subspaces,
    then for each $V\in \cE(H(0,\infty),T)$, with $T(t)=U(\tau(t))$ the 
    one-parameter group of translations, there is
    local, time-translation covariant net of standard subspaces on $M_+$
    given by
    \begin{equation*}
        K_{V}:\O \equiv I_1 \times I_2  \longmapsto K_{V}(\O):=
        \overline{H(I_1) + VH(I_2)}
        \punkt
    \end{equation*}
\end{thm}
\begin{proof}
    Isotony is obvious. 
    Locality is shown like in \cite{LoWi2010} and follows from 
    $V\in \cE(H(0,\infty))$.
    But then we have also standardness, namely $K(\O)$ is cyclic because
    $H(I_1)$ is already cyclic and separating because $K(\O)'$ 
    contains $H(I_1^>)$ where $I_1^>$ is the left component of the two piece 
    complement of $I_1$. 
    Time-translation covariance holds because $V$ commutes with $T$.
\end{proof}
\begin{cor} 
    Let $(F, \langle\slot,\slot\rangle)$ be a real $n$-dimensional Euclidean
    space and 
    $H_F$ the net of standard subspaces from 
    Proposition \ref{prop:NetStd}.
    Then for each $V\in \cE(H_F(0,\infty),T)$, 
    i.e. each element in $\cS_F$ as described 
    in 
    Remark \ref{rmk:SemigroupStd}, there is a local, time-translation 
    covariant net 
    of standard subspaces on $M_+$.
\end{cor}
\subsection{Local Nets of von Neumann Algebras on Minkowski Half-Plane}
\label{sec:BQFTNet}
Let us recall the definition of a local, time-translation covariant 
net of von Neumann algebras $\A_V$ on Minkowski half-plane \cite{LoWi2010}
related to a local Möbius covariant net $\A$ and an element of $V$ of the 
semigroup $\cE(\A)$.
\begin{defi}
    A \emph{local, time-translation covariant net (of von Neumann algebras) on Minkowski half-plane 
    (\BQFTnet)} on a Hilbert space $\Hil$ is a family 
    of von Neumann algebras $\{\cB(\cO)\}_{\cO\in\cK}$ on $\Hil$
    which fulfills:
    \begin{enumerate}[{\bf A.}]
        \item\textbf{Isotony.} $\O_1 \subset \O_2$ implies $\cB(\O_1) 
            \subset \cB(\O_2)$. 
        \item\textbf{Locality.} If  $\O_1,\O_2 \in \cK_+$ are space-like 
            separated then 
            $[\cB(\O_1),  \cB(\O_2)]=\{0\}$.
        \item\textbf{Time-translation covariance.} There is 
            a unitary continuous one-parameter group 
            $T(t)=\e^{\ima t P}$ 
            on $\Hil$ with positive generator $P$, such that 
            \begin{align*}
                T(t) \cB(\O) T(t)^\ast &= \cB(\O_t), &\O \in \cK_+
            \end{align*}
            where $\O_t = \O + (t,0)$ is the shifted double cone.
        \item\textbf{Vacuum.} There is a 
            (up to phase) unique $T$ invariant vector $\Omega \in \Hil$ 
            which is cyclic and separating for every 
            $\cB(\O)$ with $\O\in\K_+$.
    \end{enumerate}
\end{defi}
Let $\A$ be a Möbius covariant local net of von Neumann algebras 
on the Hilbert space $\Hil$, which we want to regard 
(by restriction) as a net on $\RR$. All unitaries $V$ on $\Hil$, which commutes
with the one-parameter group of translations $T(t)=U(\tau(t))$, satisfy 
$V\Omega = \Omega$ and the equivalent conditions
\begin{enumerate}
    \item $V\A(I_2)V^\ast $ commutes with $\A(I_1)$ for all intervals $I_1,I_2$ of $\RR$ such that 
        $I_2 >I_1$, i.e. $I_2$ is contained in the future of $I_1$,
    \item $V \A(a,\infty)V^\ast \subset \A(a,\infty)$ for all $a\in\RR$,
    \item $V \A(0,\infty)V^\ast \subset \A(0,\infty)$,
\end{enumerate}
form a semigroup denoted by $\cE(\A)$.
Translations $V:=T(t)\equiv U(\tau(t))$ with $t>0$ are elements in $\cE(V)$. 
Also \emph{internal symmetries} $V$ of $\A$, namely $V\A(I)V^\ast = \A(I)$ for all 
$I\in\cI$ give trivial examples of elements in $\cE(\A)$.
Besides
these trivial examples it is in general not much known if there exists
other elements, but if they exist they are of the form stated
as follows.
\begin{rmk}[\cf \cite{LoWi2010}] \label{rmk:cEA}
    Let $\A$ be a conformal net, then $\cE(\A) \subset \cE(H,T)$ with
    the one-parameter group of translations $T(t)=U(\tau(t))=\e^{\ima tP}$  and 
    the standard subspace $H=\overline{\A(0,\infty)_\mathrm{sa}\Omega}$. 
    In particular $H_0:=H\ominus\RR\Omega \subset \Hil_0$ with 
    $\Hil_0:=\Hil\ominus \CC\Omega$
    is a non-degenerated standard pair and by Theorem \ref{thm:stdred} we get  
    $V\restriction_{\Hil_0} = (\varphi_{hk}(P_0))$ 
    (by definition $V\Omega=\Omega$)    
    with $(\varphi_{hk})$ a matrix in $\cS^{(\infty)}$ and 
    $P_0=P\restriction_{\Hil_0}$, \cf \cite[Corollary 2.8]{LoWi2010}.
\end{rmk}
Let $\A$ be a conformal net and 
$V\in\cE(\A)$, then we define 
\begin{align*}
    &\A_V(\cO) := \A(I_1) \vee V\A(I_2)V^\ast, &\cO=I_1\times I_2, ~I_2>I_1
    \punkt
\end{align*}
The special case $V=1$ is exactly the conformal boundary net $\A_+$ 
defined in
\cite{LoRe2004}.
\begin{prop}[\cf \cite{LoWi2010}*{Proposition 3.3 and Corollary 3.4}] If $V\in\cE(\A)$, then $\A_V$ is a \BQFTnet. The map $\cE(\A)\ni V\longmapsto \A_V$ 
is one-to-one modulo internal symmetries,
\ie $\A_{V_1}=\A_{V_2}$ with $V_1,V_2\in \cE(\A)$ iff $V_1=V_2V$ with $V$ an 
internal symmetry.
\end{prop}
The study of such \BQFTnets $\A_V$ associated with a \CFTnet $\A$ simplifies 
therefore to the study of $\cE(\A)$. 
So the question is the characterization and classification 
of the semigroup $\cE(\A)$ for a given \CFTnet $\A$. 
The rest of the paper we investigate in the explicit construction 
of families of such elements.
\subsection{Second Quantization Boundary Nets}
\label{sec:BQFTAbNet}

Let $(F, \langle \slot,\slot\rangle)$ be a $n$-dimensional Euclidean space. 
For the net $\A_F$ of Abelian currents constructed in Section \ref{sec:ChAbNet}
we know all $V =\Gamma(V_0)\in \cE(\A_F)$ which are second quantization
unitaries by the following theorem. 

\begin{thm}[\cf {\cite[Theorem 3.6]{LoWi2010}}]
    \label{thm:semigroupAb}
    $V=\Gamma(V_0)\in \cE(\A_F)$  if and only if 
    $V_0=\varphi(P_0)$ 
    with $\varphi\in \cS_F$.
\end{thm}
\begin{rmk}
    These models are exactly the second quantization of the models constructed
    in Section \ref{sec:BQFTStdNet}.
\end{rmk}

The next task is to find which $V\in\cE(\A_F)$ extend to 
a $\tilde V \in \cE(\A_L)$ for 
an even integral lattice $L\subset F$.

\subsection{Semigroup for Subnets}
If we have a conformal net $\cB$ with conformal subnet $\A$ and $V\in\cE(\cB)$ 
the question
arises when $V$ restricts to an element in $\cE(\A)$. 
\begin{lem}\label{lem:JonesvsCondExp}
    Let $\Omega \in \Hil$ be a cyclic and separating vector for the von Neumann 
    factor
    $B\subset\B(\Hil)$ and separating for the subfactor $A \subset B$ 
    and we assume there is a conditional expectation $E_A:B\to A$ which leaves 
    the state $\phi_\Omega=(\Omega,\slot\Omega)$ invariant. 
    Let 
    $V \in \U(\Hil)$ with $V\Omega = \Omega$ and
    $VBV^\ast \subset B$. Then the following is equivalent
    \begin{enumerate}
        \item $V$ commutes with the Jones projection $e_A$.
        \item $E_A$ and $\Ad_V$ commute, \ie
            $E_A(VbV^\ast)=VE_A(b)V^\ast$ for all $b\in B$.
    \end{enumerate}
\end{lem}
\begin{proof}    
    By definition we have $E_A(b)  e_A = e_A b e_A$ for all $b\in B$. 
    Let us assume $[e_A,V]=0$, then
    \begin{align*}
        VE_A(b) V^\ast \Omega 
            &= V E_A(b) \Omega\\ 
            &= V E_A(b) e_A \Omega\\ 
            &= Ve_A b e_A V^\ast \Omega\\
            &= e_A V b V^\ast  e_A \Omega\\
            &= E_A (VbV^\ast) e_A \Omega\\
            &= E_A (VbV^\ast) \Omega
    \end{align*}
    and by the separating property $[E_A,\Ad_V]=0$ follows.
    On the other hand,  let us assume now $[E_A,\Ad_V]=0$.    Then 
    \begin{align*}
        e_A V b \Omega &= e_A V b V^\ast e_A \Omega 
        \\&= E_A(\Ad_V (b)) e_A \Omega
        \\&= \Ad_V(E_A(b)) e_A \Omega
        \\&= V E_A(b) \Omega
        \\&= V E_A(b)e_A \Omega
        \\&= V e_A b e_A \Omega
        \\&= V e_A b \Omega
    \end{align*}
    and cyclicity implies that $e_AV=Ve_A$.
\end{proof}
\begin{prop}
    \label{prop:restrictions}
    Let $\cB$ be a conformal net on $\Hil$ with vacuum $\Omega$; 
    let $\A$ be conformal subnet of $\cB$ and let $e$ be the projection on 
    $\overline{\A(I)\Omega}$ for some $I\in\cI$. 
    Further let $V \in \cE(\cB)$ and $\eta=\Ad V$, then the following are equivalent
    \begin{enumerate}
        \item $V \restriction_{e\Hil} \in \cE(\A)$, regarding $\A$ as a conformal net on $e\Hil$.
        \item For every $a\in \RR$ it is $\eta (E_a(b))=E_a(\eta(b))$ 
            for all $b\in \cB(a,\infty)$,
            where $E_a$ 
            is the conditional expectation 
            $\cB(a,\infty) \longrightarrow \A(a,\infty)$. 
        \item It is $\eta (E_0(b))=E_0(\eta(b))$ 
            for all $b\in \cB(a,\infty)$,
            where $E_0$ is the conditional expectation 
            $\cB(0,\infty) \longrightarrow \A(0,\infty)$. 
        \item \label{item:proj} 
            $V$ commutes with the projection $e$.
    \end{enumerate}
\end{prop}
\begin{proof}
    The projection $e$ does not depend on $I$ and is the Jones projection of
    the inclusion $\A(I)\subset \cB(I)$ for any $I\in \cI$. 
    Let $V\in\cE(\cB)$ such that 
    $V\restriction_{e\Hil} \in \cE(\A)$. We show that (\ref{item:proj}) 
    is true, namely for $a\in \A(0,\infty)$ 
    using $V\A(0,\infty)V^\ast \subset \A(0,\infty)$ we compute
    \begin{align*} 
         eVa \Omega   &= eVaV^\ast \Omega \\
                        &= Va V^\ast \Omega\\
                        & = V a \Omega\\
                        & = V e a \Omega
    \end{align*}
    thus by continuity $[V,e] \restriction e \Hil =0$.
    Let us write $\Hil = e\Hil \oplus e^\perp \Hil$ and let $V_1=eVe=Ve$ and 
    $V_2= e^\perp V e^\perp$; we write
    \begin{align*} 
        V=  \begin{pmatrix}
                V_1     & X \\ 
                0       & V_2 
            \end{pmatrix}
        \punkt
    \end{align*}
    Because $V$ and $V_1$ are unitaries on $\Hil$ and $e\Hil$, 
    respectively, it follows that $X=0$.
    We claim that also $[V,e] \restriction e^\perp \Hil =0$, namely for
    $\xi \in e_A^\perp \Hil$ we calculate
    $$
        e_A V\xi    = e_A V_2\xi
                    = e_A e_A^\perp V e_A^\perp \xi 
                    = 0 
                    = V e_A \xi 
        \punkt
    $$
    Conversely, let $V\in \cE(\A)$ with  $[e,V] = 0$. 
    By Lemma \ref{lem:JonesvsCondExp} $\eta=\Ad_V$ commutes with the 
    conditional expectation $E:\cB(0,\infty) \longrightarrow \A(0,\infty)$,
    \ie $E(VaV^\ast) = VE(a)V^\ast$ for $a\in\A(0,\infty)$. We claim
    that $\Ad_V$ is an endomorphism of $\A(0,\infty)$, namely
    \begin{align*}
        V\A(0,\infty)V^\ast &= VE( \cB(0,\infty))V^\ast\\
                            &= E(V\cB(0,\infty)V^\ast)\\
                            &\subset  E(\cB(0,\infty))\\
                            &=\A(0,\infty) 
        \punkt
    \end{align*}
    Since $V$ and $e$ commute $V\restriction_{e\Hil} = eVe$ is a unitary on 
    $\e\Hil$ and commutes with $T(t)\restriction_{eH}=eT(t)e$, i.e.
    $V\restriction_{e\Hil} \in \cE(\A)$.
\end{proof}
\subsection{Extensions for the Crossed Product with Free Abelian Groups}
Let $\M$ be a type \three{} factor. $\End(\M)$ is a tensor-\Cstar-category with
objects $\rho\in\End(\M)$ normal endomorphisms of $\M$ and 
arrows $\Hom_\M(\rho,\eta)=
\{t\in \M: t\rho(x)=\eta(x) t \text{ for all } x\in \M\}$. 
For any $\rho \in\End(\M)$ we have $\Hom(\rho,\rho)\ni\id_\rho:=1$. The tensor 
product is defined by the composition $\eta\otimes\rho:=\eta\rho$
and for $f\in\Hom_\M(\rho,\rho')$ and $g\in\Hom_\M(\eta,\eta')$
it is $f\otimes g := f\rho(g) = \rho'(g) f \in \Hom_\M(\rho\eta,\rho'\eta')$.

Let $L$ be a free Abelian group of rank $n$ with generators ($\ZZ$-basis) 
$\lbrace\alpha_1,\ldots\alpha_n\rbrace$ and let $\beta$ be a faithful action on
a von Neumann algebra $\M \subset \B(\Hil_0)$ with cyclic and separating vector
$\Omega$. 
The action of $L$ is characterized by the action of the automorphisms 
$\beta_i:= \beta_{\alpha_i}$. 
Let $\Hil=\oplus_{\alpha\in L}\Hil_\alpha\supset \Hil_0$. 
We assume $\beta_i$ to be implemented by unitaries $U_i$ mapping 
$\Hil_\alpha \to \Hil_{\alpha+\alpha_i}$. 
We note that
\begin{align*}
    L\ni g=\sum_{i=1}^n g_i\alpha_i &\longmapsto 
    U_g=U_1^{g_1} \cdots U_n^{g_n} & g_i\in \ZZ
\end{align*}
defines a projective representation of $L$ on $\Hil$.

Let $\tilde \M = \M \rtimes_\beta Q \subset \B(\Hil)$ be the von Neumann 
algebra generated by $\M$ and $\{U_i\}$ on $\Hil$.
We are interested in extension of endomorphisms of $\M$ to endomorphisms
of $\tilde \M$. 
The following is in principal a generalization of 
\cite[Proposition 3.8 and 3.9]{LoWi2010}.
\begin{lem}
    \label{lem:ext}
    Let $\M$ as above and $\tilde \M_0\subset\tilde \M$ the algebra finitely
    generated by $\M$ and $\{U_\alpha\}_{\alpha\in L}$.
    Further let $R:L \to L$ be an automorphism of $L$ and for $i=1,\ldots, n$
    let  $\tilde \beta_i := \beta_{R(\alpha_i)}$ be automorphisms of $\M$ 
    having $\tilde U_i=U_{R(\alpha_i)}$ as implementing unitaries.
    If there exist unitaries 
        $z_i\in \Hom_\M(\tilde\beta_i\circ\eta,\eta\circ\beta_i)$
    satisfying 
    \begin{align*} 
        z_i\tilde\beta_i(z_j) &= z_j\tilde\beta_j(z_i)
    \end{align*}
    then $\eta$ extends to an endomorphism 
    $\tilde\eta_0$ of $\tilde \M_0$ characterized by 
    $\tilde\eta_0(U_i) = z_i \tilde  U_i$. 
\end{lem}
\begin{proof} 
    Each $g\in L$ can uniquely be written as $g=\sum_i g_i\alpha_i$ with
    $g_i\in \ZZ$ for $i=1,\ldots,n$.
    Further we denote $U_g:=U_1^{g_1}\cdots U_n^{g_n}$ which defines a 
    projective representation of $Q$. For finite non-zero $a_g \in M$ we define:
    \begin{align*}
        \tilde \eta_0: \sum_g a_g U_g 
        \longmapsto & \sum_g a_g (z_1U_1)^{g_1}\cdots (z_nU_n)^{g_n} =:\sum_g a_g v_g U_g
    \end{align*}
    which is well-defined because the action is faithful. 
    It is easy to check that $\tilde\eta_0$ is an endomorphism if $v_g\in M$ 
    is a ``cocycle'' (similar like in \cite{Ka2001}) satisfying
    \begin{align*}
        v_g\tilde\beta_g(\eta(x)) &= \eta(\beta_g(x)) v_g &x\in \M
        \\v_{g+h} &= v_g\tilde\beta_g(v_h) 
    \end{align*}
    with $\tilde\beta_g=\Ad \tilde U_g$. 
    Indeed, using the tensor category calculus we write for arrows 
    $t:\sigma\eta\to\eta\sigma'$
    and $s:\rho\eta\to\eta\rho'$
    \begin{align}
        s\diamond t:=(s \otimes \id_{\tau'})(\id_\rho\otimes t) : \rho\sigma\eta \xrightarrow{\id_\rho\otimes t}
        \rho\eta\sigma' \xrightarrow{s \otimes \id_{\tau'}} \eta\rho'\tau'
    \end{align} 
    for example
    $z_i\diamond z_j:=(z_i\otimes
    \id_{\beta_j})(\id_{\tilde \beta_i}\otimes z_j) \equiv z_i\tilde\beta_i(z_j)$. The
    condition $z_i\tilde\beta_i(z_j)=z_j\tilde\beta_j(z_i)$ reads
    $z_i \diamond z_j = z_j \diamond z_i$. 
    Let us write
    $z_i^{-}=\tilde\beta_i^{-1}(z_i^\ast)$ in particular $z_i\diamond z_i^-=z_i^-\diamond z_i =
    1$ and we have also
    \begin{align*} 
        z_j\diamond z_i^- &=z_j\tilde\beta_j(\tilde\beta_i^{-1}(z_i^\ast))
        \\ &= \tilde\beta_i^{-1}(\tilde\beta_i(z_j)\tilde\beta_j(z_i^\ast))
        \\ &= \tilde\beta_i^{-1}(z_i^\ast z_i \tilde\beta_i(z_j)\tilde\beta_j( z_i^\ast))
        \\ &= \tilde\beta_i^{-1}(z_i^\ast z_j \tilde\beta_j(z_i)\tilde\beta_j( z_i^\ast))
        \\ &= \tilde\beta_i^{-1}(z_i^\ast z_j)
        \\&= z_i^-\diamond z_j
        \punkt
    \end{align*}
    With this notation it is  
    \begin{align*} 
        v_g=z_1^{\diamond g_1}\diamond\cdots\diamond z_n^{\diamond g_n} :=
        \underbrace{z_1^\pm\diamond \cdots \diamond z_1^\pm}_{\pm g_1\text{--times}} 
            \diamond \cdots \diamond
            \underbrace{z_n^\pm \diamond \cdots \diamond z_n^\pm}_{\pm g_n\text{--times}}
        \in \Hom_\M(\tilde\beta_g\eta,\eta\beta_g)
    \end{align*}
    which does not depend on the order of the $z_i$ and $z_i^-$, so in particular 
    \begin{align*}
        v_g\tilde \beta_h (v_h) \equiv  (v_g \otimes \id_{\beta_h})(\id_{\tilde\beta_g}\otimes v_h)
            &= v_g \diamond v_h = v_{g+h} \punkt&
            \qedhere
    \end{align*}
\end{proof}

\begin{prop}
    \label{prop:ext}
    Let $\eta$ be a $\phi_\Omega$-preserving endomorphism of $\M$.
    Under the hypothesis of Lemma \ref{lem:ext} the endomorphism $\eta$ extends
    to a $\phi_\Omega$-preserving endomorphism $\tilde\eta$ of $\tilde \M$ 
    characterized by $\tilde\eta(U_i) = z_i \tilde U_i$. 
\end{prop}
\begin{proof}
    $\tilde \eta_0$ preserves the conditional expectation $\sum a_\alpha
    U_\alpha
    \mapsto a_0$ so it preserves the state $\phi_\Omega$ and $\Omega$ is cyclic
    for $\tilde\eta_0(\tilde\M_0)$, because the space
    $\overline{\tilde\eta_0(\tilde\M_0)\Omega}$ contains $\Hil_0$ and is
    $U_i$ invariant. Finally, there exists a unitary 
    $\tilde V$ with $\tilde Vx\Omega=\tilde\eta_0(x)\Omega$ and 
    $\tilde\eta=\Ad \tilde V$ is the extension.
\end{proof}
Let us in the case $\eta\in\Aut(\M)$ and $v_g\in \TT$ speak of an 
\emph{internal
symmetry}. 
In the special case $z_i=1$ for $i=1,\ldots, n$ it is $\tilde\beta_i\eta=\eta\beta_i$ 
and $\eta$ extends to a symmetry $\tilde\eta$ related to the automorphism $R$ 
of $L$;
in the case $\eta=\id_\M$ we speak of a \emph{toral symmetry}.
On the other hand let's in the case $R=\id_L$ talk about \emph{charge preserving}
endomorphisms. A charge preserving internal symmetry is toral.
\begin{rmk}
    Let $\tilde \tau:U_i\longmapsto z_iU_i$ a charge preserving transformation 
    which 
    extends $\tau$ and $\tilde \sigma:U_i \to c_i\tilde U_i$ inner
    then $\widetilde{\tau\sigma}: U_i \longmapsto c_i z_i\tilde U_i$ defines 
    an extension of
    $\tau\sigma$.

    Given $\tilde\eta$ an extension of $\eta$ with $R$ and $\tilde\sigma$ an 
    inner transformation 
    with $\tilde\sigma: \tilde U_i\to c_i U_i^\ast$ where $c_i\in\TT$ extending 
    some $\sigma\in\Aut(\M)$ having 
    $R^{-1}:L\longrightarrow L$ as automorphism of
    $L$. 
    Then $\tilde\eta\tilde\tau$ is charge preserving.
\end{rmk}
\begin{rmk} In the case when $\eta$ has a charge preserving extension $\tilde
    \eta:U_i\longmapsto z_iU_i$, let us look into the full monoidal subcategory
    $\cC$ generated by $\beta_i$ and $\eta$. Then 
    $v_g \in \Hom_M(\beta_g\eta,\eta\beta_g)$ is similar (the number of
    endomorphisms is not finite) two the half-braiding with respect to $\eta$ 
    defined in \cite{Iz2000}. 
    The condition $z_i\beta_i(z_j)=z_j\beta_j(z_i)$ reflects the fusion-braid 
    equation.
\end{rmk}
We also have a converse of Proposition \ref{prop:ext}, namely that extensions of this form are given by $z_i$ 
like in Lemma \ref{lem:ext}. 
\begin{prop} \label{prop:converse}
    If $\tilde\eta$ is an endomorphism of $\tilde \M$ 
    and restricts to an endomorphism of $\M$ and
    $\eta(e_\alpha)=e_{R(\alpha)}$ such that 
    $U_iU_j = \hat c(\alpha_i,\alpha_j) U_jU_i$ $\Longleftrightarrow$ $\tilde U_i\tilde U_j = 
    \hat c (\alpha_i,\alpha_j) \tilde U_j
    \tilde U_i$.
    Then there exist $z_i\in \Hom_\M(\tilde \beta_i\eta,\eta\beta_i)$
    with $z_i\tilde\beta_i(z_j)=z_j\tilde\beta_j(z_i)$.
\end{prop}
\begin{proof}
    If $\tilde\eta$ restricts to an endomorphism of $\M$ means it 
    commutes with the Jones projection $e_0$ by Proposition 
    \ref{prop:restrictions} and  
    $z_i:=\tilde \eta(U_i)\tilde U_i^\ast \in \tilde \M$. But also 
    $z_i\in \M$ because it commutes with $e_0$. Finally
    \begin{align*}
            z_i\tilde\beta_i(\eta(x))
                &=\tilde\eta(U_i)\eta(x)\tilde U_i^\ast
            \\  &=\eta(\beta_i(x))\tilde\eta(U_i)\tilde U_i^\ast
            \\  &=\eta(\beta_i(x))z_i
    \end{align*}
    and 
    \begin{align*}
            z_i\tilde\beta_i(z_j)   
                &=\tilde\eta(U_iU_j)\tilde U_j^\ast \tilde U_i^\ast
            \\  &=\tilde\eta\left(\hat c(\alpha_i,\alpha_j)U_jU_i\right)
                \hat c(\alpha_i,\alpha_j)^\ast \tilde U_i^\ast \tilde U_j^\ast
            \\  &=\tilde\eta(U_jU_i)
                U_i^\ast \tilde U_j^\ast
            \\  &=z_j\tilde\beta_j(z_i)   
    \end{align*}
    which completes the proof.
\end{proof}
\begin{prop}\label{prop:extV}
    Let $\A$ be conformal net and $\Aext$ a local extension by
    $\lbrace\beta_i\rbrace_{i=1\ldots n}$ automorphisms of $\A$ localized in 
    $(0,\infty)$, such that $\Aext(0,\infty)= \A(0,\infty) \rtimes_\beta Q$. 
    Further let $V \in \cE(\A)$, $\eta = \Ad V$. 
    Then there exists an extension $\tilde V$ of $V$, with 
    $\tilde V \in \cE(\Aext)$ associated to an automorphism 
    $R:\alpha_i \longmapsto \tilde\alpha_i$ of $L$ if and only if
    \begin{enumerate}
        \item there are 
            $z_i\in \Hom_{\A(0,\infty)}(\tilde\beta_i\eta,\eta\beta_i)$ for 
            $i=1,\ldots,n$ such that 
            \begin{align*}
                z_i\tilde\beta_i(z_j) &=z_j\tilde\beta_j(z_i)
                \komma
            \end{align*}
        \item and there are unitary one-parameter groups $u_i(t)$ with 
            $\Ad u_i(t) \beta_i(\tau_t(x)) = \tau_t(\beta_i(x))$
            for all $x\in \A(0,\infty)$ 
            with $u_i(t)\beta_i(u_j(t))=u_j(t)\beta_j(u_i(t))$ which extends
            $\tau_t=\Ad T(t)$ from $\A$ 
            to $\Aext$ by $\tilde\tau_t(U_i)=u_i(t)U_i$ satisfying
            \begin{equation*}
                z_i \tilde u_i(t)^\ast =\eta\left(u_i(t)^\ast\right)\tau_t(z_i)
                \punkt
            \end{equation*}
    \end{enumerate}
\end{prop}

\begin{proof}
    The first part follows directly by Proposition \ref{prop:ext}
    and the converse by Proposition \ref{prop:converse}.
    We note that $\tau_t$ is extended to $\Aext$ via a cocycle
    $u_i(t)$ namely $\tilde\tau_t(U_i) = u_i(t)U_i$. 
    That $\tilde\tau_t$ commutes with $\tilde\eta$ means equality of 
    \begin{align*}
        \tilde\tau_t(\tilde\eta(U_i)) 
            &= \tilde\tau_t(z_i\tilde U_i)=\tau_t(z_i)\tilde u_i(t)\tilde U_i\\
        \tilde\eta(\tilde\tau_t(U_i)) 
            &= \tilde\eta(u_i(t)U_i) = \eta(u_i(t)) z_i\tilde U_i
    \end{align*}
    which is equivalent with
    \begin{equation*}
        \eta\left(u_i(t)^\ast\right)\tau_t(z_i)=z_i \tilde u_i(t)^\ast.
    \end{equation*}
\end{proof}
\subsection{Boundary Nets Associated with Lattices}
\label{sec:BQFTLaNets}
We investigate in the semigroup elements for the \CFTnets associated with 
lattices and give corresponding \BQFTnets. 
Here it is more convenient to use the real line picture by 
identifying $x=\tan \theta/2$.
For $f\in L^2(\RR,F)$, we denote its Fourier transform by 
$\hat f\in L^2(\RR,F_\CC)$, namely:
\begin{align*}
    \hat f(p) &= \frac1{2\pi}\int_\RR \e^{-\ima px} f(x) \dd x &
    \overline{\hat f(p)}=\hat f(-p)
    \\f(x) &= \int_\RR \e^{\ima px} \hat f(p)\dd p
    \punkt
\end{align*}
We note that in $\Hil_{0,F}$ the complex structure is given by 
$\widehat{\cJ f}(p)= -\ima \sign(p)\hat f(p)$ 
and the action of the translation by $T(t)=\e^{\ima t P}$ by 
$$  
    \widehat{T(t)f}(p)=\e^{-\ima \sign(p) t |p|}\hat f(p)=\e^{-\ima t p} 
    \hat f(p)
    \punkt
$$ 
and the sesquilinear form by
\begin{equation*} 
    \omega(f,g) := \Im(f,g)
    = \frac 1 2 
    \int_\RR \langle f'(x),g(x)\rangle\frac{\dd x }{2\pi} =:  
    \frac12\int \langle f',g\rangle
    \punkt
\end{equation*}
The norm of $\Hil_{0,F}$ is 
$\| f \|_{\Hil_{0,F}} = \mathrm{const.}\int_0^\infty \|\hat f(p)\|_{F_\CC} 
p \dd p$ and we note that $f\in L^2(\RR,F)$ is in 
$\Hil_{0,F}$ if the norm 
$\| f \|_{\Hil_{0,F}}$ is finite. 

Let $L$ be an even lattice. We write it as a sum of irreducible components
$L=L_1 \oplus \cdots \oplus L_k$ with $\langle L_i,L_j\rangle=0$.
We call a linear, isometric, isomorphic map $L\xrightarrow{\sim} L$ an 
\emph{automorphsim} of $L$ and denote the set of 
automorphisms of $L$ by $\Aut L$.
\begin{defi}
    Let $R:L\xrightarrow\sim L$ be an automorphism of $L=L_1\oplus\cdots\oplus
    L_k$ and $F=L\otimes_\ZZ\RR$. We denote by $\cS_{L,R}$ the space of 
    elements $\varphi\in\cS_F$, such that $\varphi(p)$ maps $\CC\alpha_i$ to 
    $\CC R\alpha_i$ for all $i=1,\ldots, n$ and for almost all $p$.
\end{defi}
\begin{lem}\label{lem:SF}   
    With this notation, there is a bijection between $\cS_{L,R}$  
    and $\cS^{ \times k}$. 
    It is given by 
    $\cS^{ \times k}\ni(\varphi_1,\ldots,\varphi_k)\longmapsto \varphi$ with
    \begin{align*}
        \varphi(p) \alpha_i &:= \varphi_j(p) R\alpha_i & \alpha_i \in L_j
        .
    \end{align*}
\end{lem}
\begin{proof} We write $\varphi(p) \alpha_i = c_i(p) R\alpha_i$
    with $c_i\in \cS$. 
    That $\varphi(p)\in \U(F_\CC)$ is equivalent with,
    $  
        \overline {c_i(p)}c_j(p)\langle \alpha_i,\alpha_j\rangle
        = \langle\alpha_i,\alpha_j\rangle
    $
    for all $i,j$. But this means $c_i(p)=c_j(p)$ on each component.
\end{proof}
Let us abbreviate $\tilde \alpha_i =R\alpha_i$.
We call $\varphi\in \cS$ \emph{Hölder continuous} at $0$, if 
\begin {equation*}  
    p\longmapsto \frac{|\varphi(p)-1|^2}{|p|}
\end{equation*}
is locally integrable and denote the subset of Hölder continuous 
functions by $\cS^\Hol$. In a obvious way we denote 
$\cS_{L,R}^\Hol \cong \cS^{\Hol\times k}$.
\begin{lem} Let $L$ be an even lattice, $R\in\Aut(L)$ and $F=L\otimes_\ZZ\CC$.
    Let $\eta = \Ad V$ with $V\in \cE(\A_F)$ related to  
    $\varphi\in\cS_{L,R}^\Hol\subset \cS_F$ like in 
    Theorem \ref{thm:semigroupAb}.
    Then there exist unitaries $z_i\in \A_F(0,\infty)$, such that 
    \begin{enumerate}
        \item $z_i\in\Hom_{\A_F(0,\infty)}(\tilde\beta_i\eta,\eta\beta_i)$,
            \label{item:a1}
        \item  $z_i\tilde\beta_i(z_j) = z_j\tilde\beta_j(z_i)$,
            \label{item:a2}
        \item $z_i\tilde u_i(t)^\ast = \eta\left(u_i(t)^\ast\right)\tau_t(z_i)$.
            \label{item:a3}
    \end{enumerate}
\end{lem}
\begin{proof}
The automorphisms localized in $(0,\infty)$ can be chosen to be 
\begin{align*}
    \beta_i(W(f)) &= \e^{\ima \int \langle f, \emm \cdot \alpha_i\rangle } W(f)
\end{align*} 
with $\emm:\RR \longrightarrow \RR$ a Schwartz function with 
support in $(0,\infty)$ and $\int_\RR  \emm(x)=1$.
Let $R\in\Aut(L)$ and $\varphi\in \cS_{L,R}^\Hol$ with corresponding 
$(\varphi_1,\ldots,\varphi_k)\in \cS^{\Hol\times k}$ given by Lemma 
\ref{lem:SF} and let us formally define $\emm_i := \varphi_j(P)m\tilde\alpha_j$
for $\alpha_i\in L_j$, more precisely
\begin{align*}
    \emm_{i}(x)&=
    \int \e^{\ima p x} \overline{\varphi_j(p)}\hat\emm(p)\tilde\alpha_i\dd p & 
  \alpha_i \in L_j 
  \punkt
\end{align*}
Then $\emm\tilde\alpha_i-\emm_i$ has zero integral, because 
$\hat \emm(0)\tilde\alpha_i = \hat \emm_i(0)$ and it is in $H_F(0,\infty)$ 
because $\varphi_j\in\cS$ is analytic in the upper strip 
using the Paley-Wiener theorem. 
Further its principal $\eM_i-\eM\tilde\alpha_i$ has support in $(0,\infty)$ and
is in $\Hil_{0,F}$ because the norm 
\begin{align*}
    \int_0^\infty \|\hat\eM_i(p)-\hat\eM(p)\tilde\alpha_i\|_{F_\CC}^2 
    \,p\,\dd p
    =\int_0^\infty 
    \frac{ |\varphi(p)-1|^2}{|p|} \|\hat \emm(p)\tilde\alpha_i\|_{F_\CC}^2 
    \,\dd p <\infty
\end{align*}
is finite due to the Hölder continuity. In particular, we get 
$\eM_i-\eM\tilde \alpha_i\in H_F(0,\infty)$.

We claim that $z_i := W(\eM_i-\eM\tilde\alpha_i) \in \A_F(0,\infty)$
defines unitaries with the wanted properties.
Namely, to check (\ref{item:a1}) let us calculate 
    \begin{align*}
    \Ad{z_i} ( \tilde\beta_i(\eta(W(f)))) 
        &=  \Ad{z_i} (\tilde\beta_i(W(V_0 f)))
    \\  &=  \e^{\ima \int \langle \emm \tilde\alpha_i,V_0f \rangle}
            \Ad{z_i} (W(V_0 f))
    \\  &=  \e^{\ima \int \langle \emm \tilde\alpha_i,V_0f \rangle} 
            \e^{\ima \int \langle (\eM_i - \eM\tilde\alpha_i)', V_0f\rangle}
            W(V_0 f)
    \\  &=  \e^{\ima \int \langle M_i',V_0f \rangle}
            W(V_0 f) 
    \\  &=  \e^{\ima \int \langle \emm\alpha_i,f \rangle}
            W(V_0 f)
    \\  &=  \e^{\ima \int \langle \emm \alpha_i,f \rangle}\eta(W(f))
    \\  &=  \eta(\beta_i(W(f))).
    \end{align*}
To verify (\ref{item:a2}) we compute
    \begin{align*}
        z_i\tilde\beta_i(z_j)   
            &=  W(\eM_i-\eM\tilde\alpha_i)\tilde\beta_i(W(\eM_j-\eM\alpha_j))
        \\  &=  \e^{\ima\int\langle\emm\tilde\alpha_i,
                    \eM_j-\eM\tilde\alpha_j \rangle}
                 W(\eM_i-\eM\tilde\alpha_i)W(\eM_j-\eM\tilde\alpha_j)
        \\  &=  \e^{\ima\int
                    \langle\emm\tilde\alpha_i, \eM_j-\eM \alpha_j\rangle}
                \e^{\frac\ima2\int\langle \eM_i'-\emm\tilde\alpha_i,
                    \eM_j-\eM\tilde\alpha_j\rangle}
        \\  &=  
                \e^{\frac\ima2\int\langle \eM_i'+\emm\tilde\alpha_i,
                    \eM_j-\eM\tilde\alpha_j\rangle}
                W(\eM_i+\eM_j-\eM(\tilde\alpha_i+\tilde\alpha_j))
                W(\eM_i+\eM_j-\eM(\tilde\alpha_i+\tilde\alpha_j))
    \end{align*}
    which is symmetric under $i \leftrightarrow j$ realizing that
    \begin{align*}
        &\langle M_i'+m\tilde\alpha_i,M_j-M\tilde\alpha_j\rangle
        \\&\quad
            = \langle \eM_i',\eM_j\rangle
            - \langle \eM_i',\eM\tilde\alpha_j\rangle
            + \langle \emm\tilde\alpha_i,\eM_j\rangle 
            - \langle \emm\tilde\alpha_i,\eM\alpha_j\rangle
        \\& \quad  
            = \langle \emm\tilde\alpha_i,M_j\rangle
            + \langle \emm\tilde\alpha_j,M_i\rangle
            - \langle \eM_i,\eM\tilde\alpha_j\rangle'
            + \frac12\langle \eM_i',\eM_j\rangle
            - \frac12\langle \eM\alpha_i',\eM\alpha_j\rangle'
    \end{align*}
    and noting that 
    $\langle \eM_i,\eM\tilde\alpha_j\rangle=\langle\eM_j,\eM\tilde\alpha_i
    \rangle$. This is true, because if  
    $\langle\tilde\alpha_i,\tilde\alpha_j\rangle\neq 0$, then $\alpha_i$ and 
    $\alpha_j$ are connected and sit in the same component, e.g. $L_k$
    and $\eM_i$ and $\eM_j$ are obtained both by multiplication with the 
    same function $\varphi_k$.

    To show (\ref{item:a3}) we recall that $u_i(t)=W( (\eM_t-\eM)\alpha_i)$
    and $\tilde u_i(t) = W( (\eM_t-\eM)\tilde\alpha_i)$
    \begin{align*} 
            z_i\tilde u_i(t)^\ast &=W(\eM_i-\eM\tilde\alpha_i)  W(\eM\tilde\alpha_i-\eM_t\tilde\alpha_i)\\
            &=\e^{\ima/2\int
                    \langle \eM_i',\eM\tilde\alpha_i\rangle 
                    - \langle \eM_i',\eM_t\tilde\alpha_i\rangle
                    - \langle \emm\tilde\alpha_i,\eM\tilde\alpha_i\rangle
                    + \langle \emm\tilde\alpha_i,\eM_t\tilde\alpha_i\rangle} W(M_i-M_t\tilde\alpha_i)\\
            &=\e^{\ima/2\int
                \langle \eM_{ti}',\eM_t\tilde\alpha_i\rangle 
                    - \langle \eM_i',\eM_t\tilde\alpha_i\rangle
                    - \langle \eM_{ti}',\eM_{ti}\rangle
                    + \langle \eM_i',\eM_{ti}\rangle} W(M_i-M_t\tilde\alpha_i)\\
          &=W(\eM_i-\eM_{ti})W( \eM_{ti}-\eM_t\tilde \alpha_i)      
      \\    &=\eta\left(W(\eM\alpha_i-\eM_t\alpha_i)  \right)\tau_t
            \left(  W( \eM_i-\eM\tilde\alpha_i)\right)
        \\  &=\eta\left(u_i(t)^\ast\right)\tau_t(z_i)\, ,
    \end{align*}
where $\eM_i$ as before, $\eM_t(x):=\eM(x-t)$ and $\eM_{it}(x):=\eM_i(x-t)$.
\end{proof}

\begin{rmk} In particular, 
    the theorem shows that $V\in \cE(\A_F)$ corresponding 
    to $\cS^{\Hol}_{L,R}$ extends to $\tilde V\in \cE(\A_L)$
    by Proposition \ref{prop:extV}.
    For the case of \BQFTnets we can choose $R=\id_L$
    because the obtained $\tilde V$ just differ by internal symmetries.
\end{rmk}
Putting this together we have proven.
\begin{prop}
    Let $L=L_1\oplus \cdots \oplus L_k$ be an even integral lattice with $k$ 
    components and  
    $\varphi \in \cS^{\Hol}_{L,1}$ corresponding to
    $(\varphi_1,\ldots,\varphi_k) \in \cS^{\Hol\times k}$,
    then there is a 
        local, time-translation covariant net on Minkowski half-plane 
    associated with the conformal net $\A_L$ and $\varphi$.
\end{prop}
\begin{cor}
    Let $G$ be a compact, simple, connected, simply connected, simply laced Lie 
    group and $\varphi \in \cS^{\Hol}$, then there is a 
        local, time-translation covariant net on Minkowski half-plane 
    associated with the conformal net $\A_{G,1}$ 
    (associated with the level 1 representation of $\Lp G$) 
    and $\varphi$.
    Further if $G$ is just semisimple, i.e. it is a product of $k$ simple 
    groups of 
    type A, D and E, then we obtain a such net 
    for every $(\varphi_1,\ldots,\varphi_k) \in \cS^{\Hol \times k}$. 
\end{cor}

\subsection{Further Examples Coming from the Orbifold Construction}
In this section we want to give further examples of \BQFTnets 
coming from the loop group net of  
$G=\Spin(2n)$ at level 2
using the orbifold construction.
\begin{defi}
    Let $\A$ be a conformal net on $\Hil$. Let $V:G\longrightarrow \U(\Hil)$ 
    be a faithful unitary representation of a finite group $G$ on $\Hil$. 
    It is said that  $G$ acts properly on the conformal net $\A$ if the 
    following conditions are satisfied:
    \begin{enumerate} 
        \item for each $I\in\cI$ and each $g\in G$, $\alpha_g(a):= V(g)a V(g)^\ast \in \A(I)$ for 
             all $a\in\A(I)$,
        \item for each $g\in G$ it is $V(g)\Omega = \Omega$.
    \end{enumerate}
\end{defi}
\begin{defi}
    Let $\A$ be a conformal net on $\Hil$ and let $V:G\longrightarrow \U(\Hil)$ 
    be a proper action on $\Hil$.
    Let $\Hil_0=\{x\in\Hil: V(g)x=x \text{ for all } g\in G \}$ and $P_0$ the 
    projection on $\Hil_0$. Then  $\cB(I)=\{a\in\A(I): \Ad V(g)a=a\}$ is 
    a conformal subnet and we denote by $\A^G(I) = \cB(I)P_0$ the conformal net 
    on $\Hil_0$, called the \emph{orbifold net}.
\end{defi}

We use following result from \cite{Xu2000-2} to obtain loop group net of 
$\Lp\Spin(m)$ at level 2.
By identifying $\RR^{2m}\ni (x,y)\longmapsto x+\ima y \in \CC^m$ where $x,y$ are
``column'' vectors with $m$ real entries we have the natural inclusion 
$\Lp \SU(m)_1 \times \Lp\U(1)_m \subset \Lp\Spin(2m)_1$ where $\U(1)$ acts on 
$\CC^m$ as scalars.
A further natural inclusion is given by 
$\Lp \Spin(m)_2 \subset \Lp \SU(m)_1 \subset \Lp\Spin(2m)_1$.
Let $K:=(I_m,-I_m) \in \SO(2m)$ which lifts to $\Spin(2m)$.
Then it is $KAK =\overline A$ for $A\in \SU(m)$ and $KAK=A$ for 
$A\in \Spin(2m)$.
$K$ defines a proper action of $\ZZ_2$ on $\A_{(\SU(n),1)}$.
\begin{prop}[Lemma 5.1 \cite{Xu2000-2}] 
    The loop group net $\A_{(\Spin(m),2)}$ of $\Spin(m)$ at level 2 is 
    isomorphic to the $\ZZ_2$ orbifold net $\A_{(\SU(m),1)}^{\ZZ_2}$ of the 
    level 1 loop group net  $\A_{(\SU(n),1)}$ associated with $\Lp \SU(n)$, 
    \ie 
    $\A_{(\Spin(m),2)}\cong \A_{(\SU(m),1)}^{\ZZ_2} \cong \A_{A_{m-1}}^{\ZZ_2}$.
\end{prop}
\begin{prop} Let $\varphi \in \cS^{\Hol}$, 
    then there is a 
        local, time translation covariant net on Minkowski half-plane
    associated with the loop group net $\A_{\Spin(m),2}$ of $\Spin(m)$ at 
    level 2. 
\end{prop}
\begin{proof}
    Let $L$ be the $A_{n-1}$ lattice, $F=L\otimes_\ZZ\RR$ the 
    associated Euclidean space 
    and $\eta=\Ad V$ the endomorphism $\A_F$ associated with the function 
    $\varphi(p)\cdot 1_{n-1}$. 
    We choose the special cocycle 
    $$
        z_i=\e^{\ima \int \langle \emm\alpha_i,\eM_i-\eM\alpha_i\rangle}
        W(\eM_i-\eM\alpha_i)
         =:\beta_i^{1/2}(W(\eM_i-\eM\alpha_i))
    $$ 
    similar like before which differs from the $z_i$ just by a phase and 
    denote $\tilde \eta =\Ad \tilde V$ the endomorphism of 
    $\A_L\equiv \A_{\SU(n),1}$ coming from the cocycle $z_i$.
    Let $\tau:W(f)\longmapsto W(-f)$, $U_\alpha \longmapsto c_\alpha U_{\alpha}^\ast$.
    This gives a proper action of $\ZZ_2$. 
    Finally $\eta$ and $\tau$ commute
    \begin{align*} 
        \eta(\tau(U_i))
        &= \eta(c_{\alpha_i} U_i^\ast)
        \\  &= \beta_i^{-1/2}(W(\eM_i-\eM\alpha_i)^\ast)c_{\alpha_i} U_i^\ast
        \\  &= \tau(\beta_i^{1/2}(W(\eM_i-\eM\alpha_i))U_i)
        \\  &= \tau(z_i U_i)
        \\  &= \tau(\eta(U_i))
    \end{align*}
    and $\tilde\eta$ restricts to an endomorphism 
    $\tilde\eta^\tau=\Ad \tilde V_1$ of 
    $\A_{\SU(n),1}^{\ZZ_2}=\A_{\Spin(n),2}$, because $\tilde\eta$ commutes
    with $\tau$ and therefore with the Jones projection on the fixpoint. 
    In particular, we have 
    constructed $\tilde V_1 \in \cE(\A_{\Spin(n),2})$.
\end{proof}
\section{Conclusions and Outlook}
By exploiting the explicit construction of a family of conformal nets 
containing loop group nets of simply 
laced groups at level 1, namely conformal nets
associated with lattices, we have obtained 
semigroup elements of the Longo--Witten semigroup $\cE(\A)$. These elements 
give rise to new models in BQFT, 
\ie local, time-translation covariant nets on Minkowski half-plane.

The level 1 loop group models can also be embedded in  
free Fermi nets, which could lead to different 
elements of the semigroup, coming from 
restrictions of second quantization unitaries.

It would be desirable to analyze the 
semigroup for loop group models
at higher level. These loop group nets are subnets of 
the tensor product of level 1 nets and one could 
ask if the here obtained endomorphism restricts to these subnets.
By applying the coset and orbifold construction one obtains
new nets and should also get new semigroup elements.
A simple example using the orbifold construction we have given in this paper.

Regarding the Longo--Witten semigroup $\cE(\A)$ in general, remarkable questions and applications
arise. An example is the 
mystery relation between  
elements of semigroup and integrable 
models with factorizing S-matrix \cite{ZaZa1979} 
on two dimensional
Minkowski space, constructed in the operator algebraic setting in \cite{Le2008}.
Both of them take inner symmetric (or scattering) functions as an input, but 
at the moment a deeper relation is not yet found.

Noteworthy applications of the Longo--Witten semigroup can be noticed in 
deformations of chiral conformal nets, where 
the endomorphisms $\Ad V$ associated with $V\in\cE(\A)$ bring 
deformations of chiral CQFT's on two dimensional Minkowski 
space. Particularly, in \cite{Ta2011}  the endomorphisms are used for a family of deformations of 
the \Uonenet and the Ising net which are both second quantization nets.
In this point the question that arises is if such deformations also exists
for the endomorphisms of the conformal nets 
associated with lattices (obtained in this work), or more generalllyy for any 
Longo--Witten endomorphism.

Another application could be the construction of massive models in higher dimensions
from conformal nets.
Here the idea is, basically, that the restriction of a massive free field net
to a light-ray gives a conformal net; then certain translations 
yield Longo--Witten endomorphisms.
In \cite{BiMeReWa2009} it is shown, in a field theoretic 
context, how to reconstruct the massive theory,
namely how one obtains back the scalar massive free field
from infinity copies of the $\U(1)$-current.
This idea translated back to the algebraic context 
uses one-parameter groups of the Longo--Witten semigroup. Unfortunately, 
Hölder continuity rules out the functions $\varphi_t(p)=\exp(-\ima t/p)$
that produce a one-parameter group with negative generator needed to construct
a 2D local net.
We hope to come back to this issue, and possibly new constructions of 
nets in higher dimensions from conformal nets.

\subsection*{Acknowledgements} 
I am grateful to Roberto Longo for his constant support
and many useful discussions. 
Further I would like to thank Yoh Tanimoto, Karl-Henning Rehren,
Sebastiano Carpi, 
André Henriques and
John E. Roberts
for useful discussions and hints.


\def\cprime{$'$}


\end{document}